\documentclass{amsart}

\usepackage{a4wide}
\usepackage{amssymb}
\usepackage{amsmath}
\usepackage{amsfonts}
\usepackage{amsthm}
\usepackage{graphicx}
\usepackage{epsfig}
\usepackage{longtable}
\newtheorem{thm}{Theorem}[section]
\newtheorem{lem}[thm]{Lemma}

\newtheorem{prop}[thm]{Proposition}

\theoremstyle{definition}
\newtheorem{defn}[thm]{Definition}
\newtheorem{eg}[thm]{Example}

\theoremstyle{remark}
\newtheorem*{rmk}{Remark}

\numberwithin{equation}{section}

\newcommand{\eps}{\varepsilon}

\newcommand{\DEF}{{:=}}
\newcommand{\FED}{{=:}}

\newcommand{\Cset}{\mathbb{C}}

\newcommand{\Rset}{\mathbb{R}}

\newcommand{\PT}[1]{\mathbf{#1}}

\newcommand{\re}{\mathop{\mathrm{Re}}}

\newcommand{\dd}{\,d}

\DeclareMathOperator{\sgn}{sgn}

\DeclareMathOperator{\BesselJ}{J}

\DeclareMathOperator{\gdist}{d}
\DeclareMathOperator{\gammafcn}{\Gamma}
\DeclareMathOperator{\ExpIntegralEi}{Ei}

\DeclareMathOperator{\digammafcn}{\psi}

\DeclareMathOperator{\FresnelS}{S}

\DeclareMathOperator{\sinc}{sinc}

\DeclareMathOperator{\zetafcn}{\zeta}

\DeclareMathOperator{\HyperF}{F}
\newcommand{\Hypergeom}[5]{{\sideset{_#1}{_#2}\HyperF\!\left(\substack{\displaystyle#3\\\displaystyle#4};#5\right)}}

\newcommand{\Pochhsymb}[2]{{\left(#1\right)_{#2}}}

\allowdisplaybreaks[1]

\title[Discrete Energy Asymptotics on a Riemannian circle]{Discrete Energy Asymptotics on a Riemannian circle} 
\author{ J. S. Brauchart\textasteriskcentered, D. P. Hardin\textdagger, and E. B. Saff\textdagger} 
\thanks{\noindent \textasteriskcentered The research of this author was supported by an APART-Fellowship of the Austrian Academy of Sciences. \\
\textdagger The research of these authors  was supported, in part, by the U. S. National Science Foundation under grant DMS-0808093. 
}

\date{\today}

\begin{document}

\address{J. S. Brauchart, D. P. Hardin and E. B. Saff:
Center for Constructive Approximation, 
Department of Mathematics, 
Vanderbilt University, 
Nashville, TN 37240, 
USA }
\email{Johann.Brauchart@Vanderbilt.Edu}
\email{Doug.Hardin@Vanderbilt.Edu}
\email{Edward.B.Saff@Vanderbilt.Edu}

\begin{abstract} 
We derive the complete asymptotic expansion in terms of powers of $N$ for the geodesic $f$-energy of $N$ equally spaced points on a rectifiable simple closed curve $\Gamma$ in $\Rset^p$, $p\geq2$, as $N \to \infty$. For $f$ decreasing and convex, such a point configuration minimizes the $f$-energy $\sum_{j\neq k}f(d(\mathbf{x}_j, \mathbf{x}_k))$, where $d$ is the geodesic distance (with respect to $\Gamma$) between points on $\Gamma$. Completely monotonic functions, analytic kernel functions, Laurent series, and weighted kernel functions $f$ are studied. 
Of particular interest are the geodesic Riesz potential $1/d^s$ ($s \neq 0$) and the geodesic logarithmic potential $\log(1/d)$. By analytic continuation we deduce the expansion for all complex values of $s$. 
\end{abstract}

\keywords{Discrete Energy Asymptotics, Geodesic Riesz Energy, Geodesic Logarithmic Energy, Riemannian Circle, Riemann Zeta Function, General Kernel Functions, Euler-MacLaurin Summation Formula}
\subjclass[2000]{52A40}

\maketitle


\section{Introduction}

Throughout this article, $\Gamma$ is a {\em Riemannian circle} (that is, a rectifiable simple closed curve in $\Rset^p$, $p\geq2$) with length $|\Gamma|$ and associated (Lebesgue) arclength measure $\sigma = \sigma_\Gamma$. 
We denote by $\ell(\PT{x}, \PT{y})$ the length of the arc of $\Gamma$ from $\PT{x}$ to $\PT{y}$, where $\PT{x}$ precedes $\PT{y}$ on $\Gamma$. Thus $\ell(\PT{x}, \PT{y}) + \ell(\PT{y}, \PT{x}) = \left| \Gamma \right|$ for all $\PT{x}, \PT{y} \in \Gamma$. 
The {\em geodesic distance $\gdist(\PT{x},\PT{y})$} between $\PT{x}$ and $\PT{y}$ on $\Gamma$ is given by the length of the shorter arc connecting $\PT{x}$ and $\PT{y}$, that is
\begin{equation} \label{geodesic.dist}
\begin{split}
\gdist(\PT{x},\PT{y}) \DEF \gdist_\Gamma(\PT{x},\PT{y}) 
&\DEF \min \left\{ \ell(\PT{x}, \PT{y}), \ell(\PT{y}, \PT{x}) \right\} 
= \frac{\left| \Gamma \right|}{2} - \left| \ell( \PT{x}, \PT{y} ) - \frac{\left| \Gamma \right|}{2} \right|.
\end{split}
\end{equation}
The geodesic distance between two points on $\Gamma$ can be at most $|\Gamma|/2$.

Given a lower semicontinuous function $f:[0, | \Gamma | / 2] \to \Rset\cup\{+\infty\}$, the discrete $f$-energy problem is concerned with properties of $N$ point systems $\PT{z}_{1,N}^*, \dots, \PT{z}_{N,N}^*$ on $\Gamma$ ($N\geq2$) that minimize the $f$-energy functional
\begin{equation} \label{G.f}
G_f(\PT{x}_1,\dots,\PT{x}_N) \DEF \sum_{j \neq k} f( \gdist(\PT{x}_j,\PT{x}_k) ) \DEF \mathop{\sum_{j=1}^N\sum_{k=1}^N}_{j \neq k} f ( \gdist(\PT{x}_j,\PT{x}_k) ),
\end{equation}
over all $N$ point configurations $\omega_N$ of not necessarily distinct points $\PT{x}_1, \dots, \PT{x}_N$ on $\Gamma$. 
The following result asserts that equally spaced points (with respect to arclength) on $\Gamma$ are minimal $f$-energy point configurations for a large class of functions $f$. 
\begin{prop} \label{prop:optimality}
Let $f:[0,| \Gamma | / 2] \to \Rset\cup\{+\infty\}$ be a lower semicontinuous function. 

{\rm (A)} If $f$ is convex and decreasing, then the geodesic $f$-energy of $N$ points on $\Gamma$ attains a global minimum at $N$ equally spaced points on $\Gamma$. If $f$ is strictly convex, then these are the only configurations that attain a global minimum. 

{\rm (B)} If $f$ is concave and decreasing, then the geodesic $f$-energy of $N$ points on $\Gamma$ attains a global minimum at antipodal systems $\omega_N$ with $\lceil N / 2 \rceil$ points at $\PT{p}$ and $\lfloor N / 2 \rfloor$ points at $\PT{q}$, where $\PT{p}$ and $\PT{q}$ are any pair of points on $\Gamma$ with geodesic distance $| \Gamma | / 2$. If $f$ is strictly concave, then these are the only configurations that attain a global minimum. 
\end{prop}

Part (A) of Proposition~\ref{prop:optimality} follows from a standard ``winding number argument'' that can be traced back to the work of Fejes T{\'o}th~\cite{Fe1956}. The result in the general form stated here appears explicitly in the work of M. G{\"o}tz \cite[Proposition~9]{Go2003} who uses a similar notion of ``orbits.'' For completeness, we present in Section~\ref{sec:proofs} a brief proof of Part (A).

\begin{rmk}
Alexander and Stolarsky \cite{AlSt1974} studied the discrete and continuous energy problem for continuous kernel functions $f$ on compact sets. In particular, they established the optimality of vertices of a regular $N$-gon circumscribed by a circle $\mathcal{C}_a$ of radius $a$ for various non-Euclidean metrics $\rho(\PT{x}, \PT{y})$ (including the geodesic metric) with respect to an energy functional $E_{\sigma,\lambda}(\PT{x}_1, \dots, \PT{x}_N) \DEF \sigma([\rho(\PT{x}_j,\PT{x}_k)]^\lambda)$, $0 < \lambda \leq 1$, on $\mathcal{C}_a$ where $\sigma$ is an elementary symmetric function on $\binom{n}{2}$ real variables. 
This result does not extend to the complete class of functions in Proposition~\ref{prop:optimality} and vice versa. However, both cover the generalized sum of geodesic distances problem. 
\end{rmk}

In the case of Riesz potentials we set
\begin{equation*}
f_s(x) \DEF - x^{-s}, \quad s < 0, \qquad f_{0}(x) \DEF \log ( 1 / x ), \qquad f_s(x) \DEF x^{-s}, \quad s > 0.
\end{equation*}
Then Proposition~\ref{prop:optimality}(A) asserts that equally spaced points are unique (up to translation along the simple closed curve $\Gamma$) optimal {\em geodesic $f_s$-energy} points for $s > -1$. (For $s>0$ this fact is also proved in the dissertation of S. Borodachov \cite[Lemma~V.3.1]{Bo2006}, see also \cite{Bo2009pre}.)
%
Proposition~\ref{prop:optimality}(B) shows that for $s<-1$ and $N\geq3$, antipodal configurations are optimal $f_s$-energy points, but  equally spaced points are {\bf not}.
%
%
(We remark that if {\em Euclidean} distance is used instead of geodesic distance, then the $N$-th roots of unity on the unit circle cease to be optimal $f_s$-energy points when $s < - 2$, cf. \cite{Bj1956} and \cite{BrHaSa2009}.)

For $s=-1$ in the geodesic case, equally spaced points are optimal but so are antipodal and other configurations. Fejes T{\'o}th~\cite{Fe1959} showed that a configuration on the unit circle is optimal with respect to the {\em sum of geodesic distances}\footnote{The analogue problem for the {\em sum of (Euclidean) distances} on the unit circle was also studied by Fejes T{\'o}th~\cite{Fe1956} who proved that only (rotated copies) of the $N$-th roots of unity are optimal.} ($s=-1$) if and only if the system is centrally symmetric for an even number of points and it is the union of a centrally symmetric set and a set $\{\PT{x}_1, \dots, \PT{x}_{2k+1}\}$ such that each half circle determined by $\PT{x}_j$ ($j=1,\dots,2k+1$) contains $k$ of the points in its interior for an odd number of points. (This result is reproved in \cite{Ji2008}.)
%
These criteria easily carry over to Riemannian circles. In particular, any system of $N$ equally spaced points on $\Gamma$ and any antipodal system on $\Gamma$ satisfy these criteria. 

\begin{rmk} \label{rmk:cohn.kumar}
Equally spaced points on the unit circle are also {\em universally optimal} in the sense of Cohn and Kumar \cite{CoKu2007}, that is, they minimize the energy functional $\sum_{j \neq k} f( | \PT{x}_j - \PT{x}_k |^2 )$ for any {\em completely monotonic} potential function $f$; that is, for a function $f$ satisfying $(-1)^k f^{(k)}(x)>0$ for all integers $k \geq 0$ and all $x \in [0,2]$.
\end{rmk}


To determine the leading term in the energy asymptotics it is useful to consider the continuous energy problem.
Let $\mathfrak{M}(\Gamma)$ denote the class of Borel probability measures supported on $\Gamma$. The {\em geodesic $f$-energy} of $\mu \in \mathfrak{M}(\Gamma)$ and the {\em minimum geodesic $f$-energy} of $\Gamma$ are defined, respectively, as 
\begin{equation*}
\mathcal{I}_f^g[\mu] \DEF \int \int f( \gdist(\PT{x}, \PT{y}) ) \dd \mu(\PT{x}) \dd \mu(\PT{y}), \qquad V_f^g(\Gamma) \DEF \inf\left\{ \mathcal{I}_f^g[\mu] : \mu \in \mathfrak{M}(\Gamma) \right\}.
\end{equation*}
The {\em continuous $f$-energy problem} concerns the existence, uniqueness, and characterization of a measure $\mu_\Gamma$ satisfying $V_f^g(\Gamma) = \mathcal{I}_f^g[\mu_\Gamma]$. If such a measure exists, it is called an {\em equilibrium measure on $\Gamma$}. 
\begin{prop} \label{prop:leading.term}
Let $f$ be a Lebesgue integrable lower semicontinuous function on $[0, | \Gamma |/2]$ and convex and decreasing on $(0, | \Gamma |/2]$. Then the normalized arclength measure $\sigma_\Gamma$ is an equilibrium measure on $\Gamma$ and 
\begin{equation} \label{eq:limit}
\lim_{N \to \infty} G_f(\omega_N^{(f)}) / N^2 = V_f^g(\Gamma).
\end{equation}
If, in addition, $f$ is strictly decreasing, then $\sigma_\Gamma$ is unique.
\end{prop}

The proofs of the propositions in this introduction are given in Section~\ref{sec:proofs}. 

Note that \eqref{eq:limit} provides the first term in the asymptotic expansion of $G_f(\omega_N^{(f)})$ for large $N$, that is $G_f(\omega_N^{(f)}) \sim V_f^g(\Gamma) \, N^2$ as $N\to\infty$. The goal of the present paper is to extend this asymptotic expansion to an arbitrary number of terms. The case when $\lim_{N \to \infty} G_f(\omega_N^{(f)}) / N^2 \to \infty$ as $N \to \infty$ is also studied. For a certain class of functions $f$ it turns out that the leading term is of the form $a_0 2 \zetafcn(s_0) | \Gamma |^{-s_0} N^{1+s_0}$ for some $s_0 > 1$, where $a_0 = \lim_{x \to 0^+} x^{s_0} f(x)$ is the coefficient of the dominant term in the asymptotic expansion of $f$ near the origin and $\zetafcn(s)$ is the classical Riemann zeta function. However, such a leading term might even not exist. Indeed, if the function $f$ has an essential singularity at $0$ and is otherwise analytic in a sufficiently large annulus centered at zero, then the asymptotics of the geodesic $f$-energy of equally spaced points on $\Gamma$ contains an infinite series part with rising positive powers of $N$ determined by the principal part of the Laurent expansion of $f$ at $0$. Consequently, there is no ``highest power of $N$'', see Examples~\ref{eg:ess.sing.1} and \ref{eg:ess.sing.2} below.

An outline of our paper is as follows. In Section~\ref{sec:f.energy}, the geodesic $f$-energy of equally spaced points on $\Gamma$ is investigated. In particular, completely monotonic functions, analytic kernel functions, Laurent series, and weighted kernel functions $f$ are considered. Illustrative examples complement this study. 
In Section~\ref{sec:geodesic.Riesz.s.energy}, the geodesic logarithmic energy and the geodesic Riesz $s$-energy of equally spaced points on $\Gamma$ are studied. The results are compared with their counterparts when $\gdist(\PT{\cdot}, \PT{\cdot})$ is replaced by the Euclidean metric. 
The proofs of the results are given in Section~\ref{sec:proofs}.

\section{The geodesic $f$-energy of equally spaced points on $\Gamma$}
\label{sec:f.energy}

\begin{defn} \label{def:main.general.f}
Given a kernel function $f:[0,|\Gamma|/2]\to\mathbb{C} \cup \{+\infty\}$, the {\em discrete geodesic $f$-energy} of $N$ equally spaced points $\PT{z}_{1,N}, \dots, \PT{z}_{N,N}$ on $\Gamma$ is denoted by
\begin{equation*}
\mathcal{M}(\Gamma,f;N) \DEF \sum_{j\neq k} f( \gdist(\PT{z}_{j,N},\PT{z}_{k,N}) ) = N \sum_{j=1}^{N-1} f( \gdist(\PT{z}_{j,N},\PT{z}_{N,N}) ).
\end{equation*}
\end{defn}

Set $N = 2 M + \kappa$ ($\kappa = 0, 1$). Using the fact that the points are equally spaced, it can be easily shown that
\begin{equation} \label{eq:cal.M.Gamma.f.N}
\mathcal{M}(\Gamma, f; N) = 2 N \sum_{n = 1}^{\lfloor N / 2 \rfloor} f( n \left| \Gamma \right| / N ) - \left( 1 - \kappa \right) f( \left| \Gamma \right| / 2 ) N.
\end{equation}
An essential observation is that the geodesic $f$-energy has (when expressed in terms of powers of $N$) different asymptotics for even $N$ and odd $N$. We remark that for real-valued functions $f$ a configuration of equally spaced points is optimal with respect to the geodesic $f$-energy defined in \eqref{G.f}, whenever $f$ satisfies the hypotheses of Proposition~\ref{prop:optimality}(A).
%

%
%
An application of the generalized Euler-MacLaurin summation formula (see Proposition~\ref{prop:Euler-MacLaurin.Summation} below) yields an exact formula for $\mathcal{M}(\Gamma, f; N)$ in terms of powers of $N$. The asymptotic analysis of this expression motivates the following definition.

\begin{defn} \label{def:admissible} A function $f:[0,|\Gamma|/2]\to\mathbb{C}\cup\{+\infty\}$ is called {\em admissible} if the following holds:
\begin{enumerate} \setlength{\itemsep}{3pt}
\item[(i)] $f$ has a continuous derivative of order $2p+1$ on the interval $(0, | \Gamma | / 2]$;
\item[(ii)] there exists a function $S_q(x)$ of the form $S_q(x) = \sum_{n=0}^{q} a_n \, x^{-s_n}$, where $a_n$ and $s_n$ ($n=0,\dots,q$) are complex numbers with $\re s_0 > \re s_1 > \cdots > \re s_q$ \footnote{The powers in $S_q(x)$ are principal values.} and $\re s_q + 2p > 0$ or $s_q = -2p$ such that for some $\delta>0$
\begin{enumerate}
\item $1 - \re s_q + \delta > 0$,
\item $\displaystyle \int_0^x \left\{ f(y) - S_{q}(y) \right\} \dd y = \mathcal{O}(x^{1+\delta-s_q})$ as $x \to 0^+$,
\item $\displaystyle \left\{ f(x) - S_{q}(x) \right\}^{(\nu)} = \mathcal{O}(x^{\delta-s_q-\nu})$ as $x \to 0^+$ for all $\nu = 0, 1, \dots, 2p+1$. 
\end{enumerate}
\end{enumerate}
\end{defn}
 
For $p\geq1$ an integer the following sum arises in the main theorems describing the asymptotics of $\mathcal{M}(\Gamma, f; N)$:
\begin{equation}
\mathcal{B}_p(\Gamma, f; N) \DEF \frac{2}{\left| \Gamma \right|} N^2 \sum_{n = 1}^{p} \frac{B_{2n}(\kappa/2)}{(2n)!} \left( \left| \Gamma \right| / N \right)^{2n} f^{(2n-1)}( \left| \Gamma \right| / 2 ), \qquad N = 2M + \kappa, \ \kappa = 0, 1, \label{eq:B.p}
\end{equation}
where $B_m(x)$ denotes the Bernoulli polynomial of degree $m$ defined by 
\begin{equation*}
\frac{z}{e^z-1} e^{x z} = \sum_{m=0}^\infty \frac{B_m(x)}{m!} \, z^m, \qquad B_m(x) = \sum_{k=0}^m \binom{m}{k} B_{m-k} x^k,
\end{equation*}
where $B_0=1$, $B_1=-1/2$, \dots are the so-called {\em Bernoulli numbers}. Recall that $B_{2k+1}=0$, $(-1)^{k-1}B_{2k}>0$ for $k=1,2,3,\dots$, and $B_{n}(1/2) = ( 2^{1-n} - 1 ) B_n$ for $n\geq0$ (\cite{AbSt1992}).


\begin{thm}[general case] \label{thm:general.f.general.case}
Let $f$ be admissible in the sense of Definition~\ref{def:admissible} and suppose none of $s_0, s_1, \dots, s_q$ equals $1$. Then, for $N = 2 M + \kappa$ with $\kappa = 0$ or $\kappa = 1$,
\begin{equation} \label{eq:general.case.asymptotics}
\mathcal{M}(\Gamma, f; N) = V_f(\Gamma) \, N^2 + \sum_{n=0}^q a_n \frac{2 \zetafcn(s_n)}{\left| \Gamma \right|^{s_n}} N^{1+s_n} + \mathcal{B}_p(\Gamma, f; N) + \mathfrak{R}_p(\Gamma, f; N),
\end{equation}
where 
\begin{equation} \label{eq:general.case.V.f}
V_f(\Gamma) = \frac{2}{\left| \Gamma \right|} \sum_{n=0}^q a_n \frac{\left( \left| \Gamma \right| / 2 \right)^{1-s_n}}{1-s_n} + \frac{2}{\left| \Gamma \right|} \int_0^{\left| \Gamma \right| / 2} ( f - S_q )(x) \dd x
\end{equation}
and the remainder term satisfies $\mathfrak{R}_p(\Gamma, f; N) = \mathcal{O}( N^{1-2p} ) + \mathcal{O}( N^{1-\delta+s_q} )$ as $N\to\infty$ if $2p \neq \delta - \re s_q$, whereas $\mathfrak{R}_p(\Gamma, f; N) = \mathcal{O}( N^{1-2p} \log N)$ if $2p = \delta - \re s_q$. 
\end{thm}

The next result involves the {\em Euler-Mascheroni constant} defined by
\begin{equation*}
\gamma \DEF \lim_{n\to\infty} \left( 1 + \frac{1}{2} + \frac{1}{3} + \frac{1}{4} + \cdots + \frac{1}{n} - \log n \right).
\end{equation*}

\begin{thm}[exceptional case] \label{thm:general.f.exceptional.case}
Let $f$ be admissible in the sense of Definition~\ref{def:admissible} and $s_{q^\prime}=1$ for some $1 \leq q^\prime \leq q$.\footnote{By Definition~\ref{def:admissible} there is only one such $s_{q^\prime}$.} Then, for $N = 2 M + \kappa$ with $\kappa = 0$ or $\kappa = 1$,
\begin{align*}
\mathcal{M}(\Gamma, f; N) 
&= \frac{2}{\left| \Gamma \right|} a_{q^\prime} \, N^2 \log N + V_f(\Gamma) \, N^2 + \sum_{\substack{n=0,\\ n\neq q^\prime}}^q a_n \frac{2 \zetafcn(s_n)}{\left| \Gamma \right|^{s_n}} N^{1+s_n} + \mathcal{B}_p(\Gamma, f; N) + \mathfrak{R}_p(\Gamma, f; N),
\end{align*}
where
\begin{equation} \label{eq:except.case.V.f}
V_f(\Gamma) = \frac{2}{\left| \Gamma \right|} \Bigg\{ \sum_{\substack{n=0,\\ n\neq q^\prime}}^q a_n \frac{\left( \left| \Gamma \right| / 2 \right)^{1-s_n}}{1-s_n} + \int_0^{\left| \Gamma \right| / 2} ( f - S_q )(x) \dd x - a_{q^\prime} \left( \log 2 - \gamma \right) \Bigg\}
\end{equation}
and the remainder term satisfies $\mathfrak{R}_p(\Gamma, f; N) = \mathcal{O}( N^{1-2p} ) + \mathcal{O}( N^{1-\delta+s_q} )$ as $N\to\infty$ if $2p \neq \delta - \re s_q$, whereas $\mathfrak{R}_p(\Gamma, f; N) = \mathcal{O}( N^{1-2p} \log N)$ if $2p = \delta - \re s_q$. 
\end{thm}

\begin{rmk}
Both Theorems \ref{thm:general.f.general.case} and \ref{thm:general.f.exceptional.case} show that only the coefficients of the nonpositive even powers of $N$ depend on the parity of $N$. These dependencies appear in the sum $\mathcal{B}_p(\Gamma, f; N)$.
\end{rmk}

%
%

%
\begin{rmk}
If $f(z) \equiv S_q(z) = \sum_{n=0}^q a_n z^{-s_n}$ for some $q$ and $\re s_0 > \cdots > \re s_q$, then all expressions in Theorems~\ref{thm:general.f.general.case} and \ref{thm:general.f.exceptional.case} containing $f-S_q$ vanish. In general, the remainder term $\mathfrak{R}_p(\Gamma, f; N)$ is of order $\mathcal{O}( N^{1-2p} )$, where the integer $p$ satisfies $\re s_q + 2p >0$. 
In particular, this holds for the Riesz kernels (cf. Theorems~\ref{thm:main} and \ref{thm:s.EQ.1} below).
%
\end{rmk}

\subsection*{Completely monotonic functions}
A non-constant {\em completely monotonic} function $f:(0,\infty) \to \mathbb{R}$ has derivatives of all orders and satisfies $(-1)^k f^{(k)}(x) > 0$ (cf. \cite{Du1940}).\footnote{A completely monotonic function on $(0,\infty)$ is necessarily analytic in the positive half-plane (\cite{Wi1946}).}
In particular, it is a continuous strictly decreasing convex function. Therefore, by Proposition~\ref{prop:optimality}, equally spaced points are optimal $f$-energy configurations on the Riemannian circle $\Gamma$. 


By Bernstein's theorem \cite[p.~161]{Wi1946} a function is completely monotonic on $(0,\infty)$ if and only if it is the Laplace transformation $f(x) = \int_0^\infty e^{-x t} \dd \mu(t)$ of some nonnegative measure $\mu$ on $[0,\infty)$ such that the integral converges for all $x>0$. 

The following result applies in particular to completely monotonic functions.
\begin{thm} \label{thm:completely.monotonic}
Let $f$ be the Laplace transform $f(x) = \int_0^\infty e^{-x t} \dd \mu(t)$ for some signed Borel measure $\mu$ on $[0,\infty)$ such that $\int_0^\infty t^m \dd |\mu|(t)$, $m = 0, 1, 2, \dots$, are all finite. Then for all integers $p\geq1$ and $N = 2M + \kappa$ with $\kappa=0, 1$
\begin{equation*}
\mathcal{M}(\Gamma, f; N) = \left\{ \frac{2}{\left| \Gamma \right|} \int_0^{\infty} \frac{1-e^{-t \left| \Gamma \right|/2}}{t} \dd \mu(t) \right\} N^2 + \sum_{n=0}^{2p} (-1)^n \frac{\mu_n}{n!} \frac{2\zetafcn(-n)}{\left| \Gamma \right|^{-n}} N^{1-n} + \mathcal{B}_p( \Gamma, f; N) + \mathcal{O}(N^{1-2p}), 
\end{equation*}
where $\mu_m \DEF \int_0^\infty t^m \dd \mu(t)$ denotes the $m$-th moment of $\mu$.
\end{thm}


\begin{rmk}
The derivation of the (complete) asymptotic expansion for $\mathcal{M}(\Gamma, f;N)$ as $N\to\infty$ for Laplace transforms for which not all moments $\mu_m$ are finite, depends on more detailed knowledge of the behavior of $f(x)$ near the origin. For example, for integral transforms $G(x) = \int_0^\infty h(x t) g(t) \dd t$ there is a well-established theory of the asymptotic expansion of $G(x)$ at $0^+$. See, \cite{HaLe1970}, \cite{HaLe1971}, \cite{BeHa1975} or \cite{Lo2008} and \cite{Du1979}. These expansions give rise to results similar to our theorem above.
\end{rmk}

\begin{rmk}
Recently, Koumandos and Pedersen~\cite{KoPe2009} studied so-called {\em completely monotonic functions of integer order $r\geq0$}, that is functions $f$ for which $x^r f(x)$ is completely monotonic. The completely monotonic functions of order $0$ are the classical completely monotonic functions; those of order $1$ are the so-called {\em strongly completely monotonic functions} satisfying that $(-1)^k x^{k+1} f^{(k)}(x)$ is nonnegative and decreasing on $(0,\infty)$. In \cite{KoPe2009} it is shown that $f$ is completely monotonic of order $\alpha>0$ ($\alpha$ real) if and only if $f$ is the Laplace transformation of a fractional integral of a positive Radon measure on $[0,\infty)$; that is 
\begin{equation*}
f(x) = \int_0^\infty e^{-x t} \mathcal{J}_\alpha[\mu](t) \dd t, \qquad \mathcal{J}_\alpha[\mu](t) \DEF \frac{1}{\Gamma(\alpha)} \int_0^t \left( t - s \right)^{\alpha-1} \dd \mu(s).
\end{equation*}
Results similar to Theorem~\ref{thm:completely.monotonic} hold for these kinds of functions. However, the problem of giving an asymptotic expansion of $f(x)$ near the origin is more subtle.
\end{rmk}

\subsection*{Analytic kernel functions} 
If $f$ is analytic in a disc with radius $| \Gamma | / 2 + \eps$ ($\eps > 0$) centered at the origin, then $f$ is admissible in the sense of Definition~\ref{def:admissible} and we have the following result.
%

\begin{thm} \label{thm:analytic.f}
Let $f(z) = \sum_{n=0}^\infty a_n z^n$ be analytic in $|z| < | \Gamma | / 2 + \eps$, $\eps>0$. Then for $N = 2 M + \kappa$ with $\kappa = 0$ or $\kappa = 1$
\begin{equation*}
\mathcal{M}(\Gamma, f; N) = \left\{ \frac{2}{\left| \Gamma \right|} \int_0^{\left| \Gamma \right| / 2} f(x) \dd x \right\} N^2 + \sum_{n=0}^{2p} a_n \frac{2 \zetafcn(-n)}{\left| \Gamma \right|^{-n}} N^{1-n} + \mathcal{B}_p(\Gamma, f; N) + \mathcal{O}_{p,|\Gamma|,f}(N^{1-2p}).
\end{equation*}
Note that $\zetafcn(0) = -1/2$ and $\zetafcn(-2k) = 0$ for $k = 1, 2, 3, \dots$.
\end{thm}

\begin{eg}
If $f(x) = e^{-x}$, then for any positive integer $p$:
\begin{align*}
\begin{split}
\mathcal{M}(\Gamma, f; N) 
&= \frac{2}{\left| \Gamma \right|} \left( 1 - e^{-\left| \Gamma \right|/2} \right) N^2 - N + \sum_{n=1}^{p} \frac{1}{(2n-1)!} \frac{2 \zetafcn(1-2n)}{\left| \Gamma \right|^{1-2n}} N^{2-2n} \\
&\phantom{=}- \sum_{n=1}^p \frac{B_{2n}(\kappa/2)}{(2n)!} \frac{2 e^{- \left| \Gamma \right| / 2}}{\left| \Gamma \right|^{1-2n}} N^{2-2n} + \mathcal{O}_{p,|\Gamma|,f}(N^{1-2p})
\end{split}
\end{align*}
as $N = 2 M + \kappa \to \infty$, where the notation of the last term indicates that the $\mathcal{O}$-constant depends on $p,|\Gamma|$ and $f$. Since $f(x)$ is a strictly decreasing convex function, by Proposition~\ref{prop:optimality}(A), equally spaced points are also optimal $f$-energy points. Thus, the relation above gives the complete asymptotics for the optimal $N$-point geodesic $e^{-(\cdot)}$-energy on Riemannian circles.
\end{eg}

\subsection*{Laurent series kernels} If $f(z)$ is analytic in the annulus $0 < |z| < | \Gamma | / 2 + \eps$ ($\eps > 0$) with a pole at $z=0$, then $f$ is admissible in the sense of Definition~\ref{def:admissible}  and we obtain the following result.

\begin{thm} \label{thm:Laurent.series}
Let $f$ be analytic in the annulus $0 < |z| < | \Gamma | / 2 + \eps$ ($\eps > 0$) having there the Laurent series expansion $f(z) = \sum_{n=-K}^\infty a_n z^n$, $K \geq 1$. 

{\rm (i)} If the residue $a_{-1}=0$, then for $N = 2 M + \kappa$ with $\kappa = 0, 1$
\begin{equation*}
\mathcal{M}(\Gamma, f; N) = V_f(\Gamma) \, N^2 + \sum_{\substack{n=-K, \\ n \neq -1}}^{2p} a_n \frac{2 \zetafcn(-n)}{\left| \Gamma \right|^{-n}} N^{1-n} + \mathcal{B}_p(\Gamma, f; N) + \mathcal{O}_{p,|\Gamma|,f}(N^{1-2p}),
\end{equation*}
where the $N^2$-coefficient is 
\begin{equation*}
V_f(\Gamma) = \frac{2}{\left| \Gamma \right|} \sum_{n=-K}^{\infty} a_n \frac{\left( \left| \Gamma \right| / 2 \right)^{1+n}}{1+n}.
\end{equation*}

{\rm (ii)} If the residue $a_{-1} \neq 0$, then for $N = 2 M + \kappa$ with $\kappa = 0, 1$
\begin{align*}
\mathcal{M}(\Gamma, f; N) 
&= \frac{2}{\left| \Gamma \right|} a_{-1} \, N^2 \log N + V_f(\Gamma) \, N^2 + \sum_{\substack{n=-K,\\ n\neq -1}}^{2p} a_n \frac{2 \zetafcn(-n)}{\left| \Gamma \right|^{-n}} N^{1-n} + \mathcal{B}_p(\Gamma, f; N) + \mathcal{O}_{p,|\Gamma|,f}(N^{1-2p}),
\end{align*}
where the $N^2$-coefficient is 
\begin{equation*}
V_f(\Gamma) = \frac{2}{\left| \Gamma \right|} \Bigg\{ \sum_{\substack{n=-K,\\ n\neq -1}}^\infty a_n \frac{\left( \left| \Gamma \right| / 2 \right)^{1+n}}{1+n} - a_{-1} \left( \log 2 - \gamma \right) \Bigg\}.
\end{equation*}
\end{thm}

Next, we give two examples of kernels $f$ each having an essential singularity at $0$. Such kernels can also be treated in the given framework, since they satisfy an extended version of Definition~\ref{def:admissible}; see Proof of Examples~\ref{eg:ess.sing.1} and \ref{eg:ess.sing.2} in Section~\ref{sec:proofs}.  

\begin{eg} \label{eg:ess.sing.1}
Let $f(x) = e^{1/x} = \sum_{n=0}^\infty 1/ (n! x^n)$, $x\in(0,+\infty)$, $f(0)=+\infty$. We define the entire function
\begin{equation*}
F(z) \DEF \sum_{n=2}^\infty \frac{\zetafcn(n)}{n!} z^n = - \gamma z - \frac{1}{2\pi i} \oint_{|w|=\rho<1} e^{z/w} \digammafcn(1-w) \dd w, \qquad z \in \mathbb{C},
\end{equation*}
where $\digammafcn(z)$ denotes the digamma function and we observe that, because of $0 < \zetafcn(n) - 1 < c 2^{-n}$ for all integers $n \geq 2$ for some $c>0$, 
\begin{equation*}
F(x) = e^x - 1 - x + \sum_{n=2}^\infty \frac{\zetafcn(n)-1}{n!} x^n = e^x + \mathcal{O}(e^{x/2}) \qquad \text{as $x\to \infty$.}
\end{equation*}
Then 
\begin{equation*}
\begin{split}
\mathcal{M}(\Gamma, f; N) &=  2 N F(N/\left|\Gamma\right|) + \frac{2}{\left| \Gamma \right|} \, N^2 \log N + V_f(\Gamma) \, N^2 - N \\
&\phantom{=}+ \sum_{n = 1}^{p} \frac{2 B_{2n}(\kappa/2)}{(2n)! \left| \Gamma \right|^{1-2n}} N^{2-2n} f^{(2n-1)}( \left| \Gamma \right| / 2 ) + \mathcal{O}_{p,|\Gamma|,f}(N^{1-2p}),
\end{split}
\end{equation*}
where
\begin{align*}
V_f(\Gamma) 
&= 1 + \frac{2}{\left| \Gamma \right|} \sum_{n=2}^\infty \frac{1}{n!} \frac{\left( \left| \Gamma \right| / 2 \right)^{1-n}}{1-n} - \frac{2}{\left| \Gamma \right|} \left( \log 2 - \gamma \right) \\
&= e^{2/\left| \Gamma \right|} - \frac{2}{\left| \Gamma \right|} \left\{ 1 - 2 \gamma + \log \left| \Gamma \right| + \ExpIntegralEi(2 / \left| \Gamma \right|) \right\},
\end{align*}
where $\ExpIntegralEi(x) = - \int_{-x}^\infty e^{-t} t^{-1} \dd t$ is the exponential integral (taking the Cauchy principal value of the integral). In particular it follows that
\begin{equation*}
\lim_{N\to\infty} \frac{\mathcal{M}(\Gamma, f; N)}{N \, e^{N/\left|\Gamma\right|}} = 2.
\end{equation*}
Since $f$ is a strictly decreasing convex function on $(0,\infty)$, by Proposition~\ref{prop:optimality}(A), equally spaced points are also optimal. Thus, the above expansion gives the asymptotics of the optimal $N$-point $e^{1/(\cdot)}$-energy.
\end{eg}

\begin{eg} \label{eg:ess.sing.2}
Let $\BesselJ_k(\lambda) = (-1)^k \BesselJ_{-k}(\lambda) \DEF \frac{1}{2\pi} \int_0^{2\pi} \cos( k \theta - \lambda \sin \theta ) \dd \theta$ denote the {\em Bessel function of the first kind of order $k$} whose generating function relation is given by (cf. \cite[Exercise~5.5(10)]{SaSn1993})
\begin{equation*}
f(x) = \exp\left[ \frac{\lambda}{2} \left( x - \frac{1}{x} \right) \right] = \sum_{n=-\infty}^\infty \BesselJ_n(\lambda) x^n \qquad \text{for  $|x|>0$.}
\end{equation*}
For integers $m\geq2$  we define the entire functions
\begin{align*}
F_m(z) &\DEF \sum_{n=m}^\infty \BesselJ_{-n}(\lambda) \zetafcn(n) z^n = \sum_{k=1}^\infty G_m(z/k), \quad G_m(z) \DEF \sum_{n=m}^\infty \BesselJ_{-n}(\lambda) z^n, \qquad z \in \mathbb{C}.
\end{align*}
If $\lambda$ is a zero of the Bessel function $\BesselJ_{-1}$, then for positive integers $p$ and $m$ $\geq2$ there holds
\begin{equation*}
\begin{split}
\mathcal{M}(\Gamma, f; N) &= 2 N F_m( N / \left| \Gamma \right| ) + 2 \sum_{n=2}^{m-1} \BesselJ_{-n}(\lambda) \zetafcn(n) \left| \Gamma \right|^{-n} N^{1+n} + V_f(\Gamma) \, N^2 + \left| \Gamma \right| B_2(\frac{\kappa}{2}) f^\prime( \left| \Gamma \right| / 2 ) \\
&\phantom{=}+ \sum_{n=2}^p \left\{ \frac{2 B_{2n}}{2n} \frac{f^{2n-1}(\left| \Gamma \right| / 2)}{(2n-1)!} + 2 \BesselJ_{2n-1}(\lambda) \zetafcn(1-2n) \right\} \left| \Gamma \right|^{2n-1} N^{2-2n} + \mathcal{O}(N^{1-2p})
\end{split}
\end{equation*}
where
\begin{equation*}
V_f(\Gamma) = \frac{2}{\left| \Gamma \right|} \sum_{\substack{n=-\infty, \\ n \neq \pm 1}}^{\infty} \BesselJ_n(\lambda) \frac{\left( \left| \Gamma \right| / 2 \right)^{1+n}}{1+n}.
\end{equation*}
If, in addition, $\lambda<0$, then $f(x)$ is a strictly decreasing convex function and, therefore, $\mathcal{M}(\Gamma, f; N)$ is also the minimal $N$-point $f$-energy on $\Gamma$ and it follows from the observation 
\begin{equation*}
G_m(x/k) = \exp\left[ - \frac{\lambda}{2} \left( \frac{x}{k} - \frac{k}{x} \right) \right] - \sum_{n=-\infty}^{m-1} J_n(\lambda) (-x/k)^n, \qquad k = 1, 2, 3, \dots,
\end{equation*}
that
\begin{equation*}
\lim_{N\to\infty} \frac{\mathcal{M}(\Gamma, f; N)}{N f(-N/|\Gamma|)} = 2.
\end{equation*}
If $\lambda$ is not a zero of $\BesselJ_{-1}$, then the above asymptotics must be modified to include a logarithmic term.
\end{eg}

%
%
\subsection*{The {\em weighted} kernel function $f_s^w(x) = x^{-s} w(x)$} 
Given a weight function $w(x)$, the kernel $f_s^w(x) = x^{-s} w(x)$ gives rise to the so-called {\em geodesic weighted Riesz $s$-energy} of an $N$-point configuration $(\PT{x}_1, \dots, \PT{x}_N)$
\begin{equation*}
G_s^w(\PT{x}_1, \dots, \PT{x}_N) \DEF \sum_{j \neq k} \frac{w(\gdist(\PT{x}_j,\PT{x}_k))}{\left[ \gdist(\PT{x}_j,\PT{x}_k) \right]^{s}}.
\end{equation*}
For the Euclidean metric the related weighted energy functionals are studied in \cite{BoHaSa2008}.

If $w(x)$ is such that $f_s^w(x)$ is admissible in the sense of Definition~\ref{def:admissible}, then Theorems~\ref{thm:general.f.general.case} and \ref{thm:general.f.exceptional.case} provide asymptotic expansions for the weighted geodesic Riesz $s$-energy of equally spaced points on a Riemannian circle $\Gamma$, which are also optimal configurations if $f_s^w(x)$ is strictly decreasing and convex (cf. Proposition~\ref{prop:optimality}(A)).

\begin{thm} \label{thm:weighted.f}
Let $w(z) = \sum_{n=0}^\infty a_n z^n$ be analytic in $|z| < | \Gamma | / 2 + \eps$, $\eps>0$. Set $f_s^w(z) \DEF z^{-s} w(z)$. Then for integers $p,q>0$ and $s\in\mathbb{C}$, $s$ not an integer, such that $q-2p < \re s < 2 + q$ we have
\begin{equation*}
\mathcal{M}(\Gamma, f_s^w; N) = V_{f_s^w}(\Gamma) \, N^2 + \sum_{n=0}^{q} a_n \frac{2 \zetafcn(s-n)}{\left| \Gamma \right|^{s-n}} N^{1+s-n} + \mathcal{B}_p(\Gamma, f_s^w; N) + \mathfrak{R}_p(\Gamma, f_s^w; N),
\end{equation*}
where $\mathcal{B}_p$ is defined in \eqref{eq:B.p}.
The $N^2$-coefficient is the meromorphic continuation to $\mathbb{C}$ of the geodesic $f_s^w$-energy of $\Gamma$ given by $( 2 / | \Gamma | ) \int_0^{| \Gamma | / 2} f_s^w(x) \dd x$ for $0 < s < 1$; that is
\begin{equation*}
V_{f_s^w}(\Gamma) = \frac{2}{\left| \Gamma \right|} \sum_{n=0}^{\infty} a_n \frac{\left( \left| \Gamma \right| / 2 \right)^{1+n-s}}{1+n-s}, \qquad s \neq 1, 2, 3, \dots.
\end{equation*}
The remainder $\mathfrak{R}_p(\Gamma, f_s^w; N)$ is of order $\mathcal{O}( N^{1-2p} ) + \mathcal{O}( N^{s-2p} )$ as $N\to\infty$. 
\end{thm}

\begin{rmk}
For $s$ is a positive integer the series $\sum_{n=0}^\infty a_n z^{n-s}$ is the Laurent expansion of $f(z)$ in $0 < |z| < | \Gamma | / 2 + \eps$ and Theorem~\ref{thm:Laurent.series} applies. For $s$ is a non-positive integer the series $\sum_{n=0}^\infty a_n z^{n-s}$ is the power series expansion of $f(z)$ in $0 < |z| < | \Gamma | / 2 + \eps$ and Theorem~\ref{thm:analytic.f} applies.
\end{rmk}

\begin{eg}
Let $w(z) = \sin( z \pi / | \Gamma | )$. Then for $\re s>0$ not an integer
\begin{equation*}
f_s^w(z) = x^{-s} w(z) = \sum_{n=0}^\infty \frac{(-1)^n}{(2n+1)!} \left( \pi / \left| \Gamma \right| \right)^{2n+1} z^{2n+1-s}
\end{equation*}
and, by Theorem~\ref{thm:weighted.f}, the geodesic weighted Riesz $s$-energy of $N$ equally spaced points has the asymptotic expansion  ($0 < \re s < 1 + 2p$)
\begin{equation*}
\mathcal{M}(\Gamma, f_s^w; N) = V_{f_s^w}(\Gamma) \, N^2 + \left( \pi / \left| \Gamma \right| \right)^{s} \sum_{k=1}^{p} \frac{(-1)^{k-1}}{(2k-1)!}  \frac{2 \zetafcn(1+s-2k)}{\pi^{1+s-2k}} N^{2+s-2k} + \mathcal{B}_p(\Gamma, f_s^w; N) + \mathfrak{R}_p(\Gamma, f_s^w; N),
\end{equation*}
where $\mathcal{B}_p(\Gamma, f_s^w; N)$ is given in \eqref{eq:B.p}. The remainder $\mathfrak{R}_p(\Gamma, f_s^w; N)$ is of order $\mathcal{O}( N^{1-2p} ) + \mathcal{O}( N^{s-2p} )$ as $N \to \infty$ and
\begin{equation*}
V_{f_s^w}(\Gamma) = \frac{2}{\pi} \left( \left| \Gamma \right| / \pi \right)^{-s} \sum_{k=1}^{\infty} \frac{(-1)^{k-1}}{(2k-1)!} \frac{\left( \pi / 2 \right)^{2k-s}}{2k-s}.
\end{equation*}
For $0 < s < 1$ we have
\begin{equation*}
V_{f_s^w}(\Gamma) = \frac{2}{\left| \Gamma \right|} \int_0^{\left| \Gamma \right| / 2} f_s^w(x) \dd x = \frac{\pi}{2} \frac{\left( \left| \Gamma \right| / 2 \right)^{-s}}{2-s} \Hypergeom{1}{2}{1-s/2}{2-s/2,3/2}{-\left( \pi / 4 \right)^2}
\end{equation*}
expressed in terms of a generalized $\Hypergeom{1}{2}{×}{×}{×}$-hypergeometric function, which is analytic at $s$ not an even integer. Hence, 
$V_{f_s^w}(\Gamma)$ is the meromorphic continuation to the complex plane of the integral $\frac{2}{\left| \Gamma \right|} \int_0^{\left| \Gamma \right| / 2} f_s^w(x) \dd x$. We observe that for $s=1/2$ we have $V_{f_s^w}(\Gamma) = 2 \sqrt{ 2 / | \Gamma | } \FresnelS(1)$, where $\FresnelS(u)$ is the Fresnel integral $\FresnelS(u) \DEF \int_0^u \sin( x^2 \pi / 2 ) \dd x$.
\end{eg}

As an application of the theorems of this section, we recover results recently given in \cite{BrHaSa2009} regarding the complete asymptotic expansion of the Euclidean Riesz $s$-energy $\mathcal{L}_s(N)$ of the $N$-th roots of unity on the unit circle $\mathbb{S}^1$ in the complex plane $\mathbb{C}$. Indeed, if $| z - w |$ denotes the Euclidean distance between two points $\zeta$ and $z$ in $\mathbb{C}$, then from the identities $| z - \zeta |^2 = 2 ( 1 - \cos \psi ) = 4 [ \sin ( \psi / 2) ]^2$, where $\psi$ denotes the angle ``between'' $\zeta$ and $z$ on $\mathbb{S}^1$, we obtain the following relation between Euclidean and geodesic Riesz $s$-kernel:
\begin{equation*} 
\left| z - \zeta \right|^{-s} = \left| 2 \left( 1 - \cos \psi \right) \right|^{s/2} = \left| 2 \sin \frac{\psi}{2} \right|^s = \left| 2 \sin \frac{\gdist(\zeta,z)}{2} \right|^s, \qquad \zeta,z \in \mathbb{S}^1.
\end{equation*}
Thus, for $\zeta,z \in \mathbb{S}^1$ there holds
\begin{equation}
\left| z - \zeta \right|^{-s} = f_s^w(\gdist(\zeta,z)), \qquad w(x) \DEF \left( \sinc \frac{x}{2} \right)^{-s}, \qquad f_s^w(x) = x^{-s} \sinc^{-s}(x/2),
\end{equation}
where the ``sinc'' function, defined as $\sinc z = ( \sin z ) / z$ is an entire function that is non-zero for $|z|<\pi$ and hence, has a logarithm $g(z)= \log \sinc z$ that is analytic for $|z|<\pi$ (we choose the branch such that $\log \sinc 0=0$). The function $\sinc^{-s} (z/2) \DEF \exp[-s \log \sinc  (z/2)]$ is even and analytic on the unit disc $|z|<2\pi$ and thus has a power series representation of the form
\begin{equation*} 
\sinc^{-s} (z/2) = \sum_{n=0}^\infty \alpha_n(s) z^{2n}, \quad |z|<2\pi, \, s\in \Cset. 
\end{equation*}
It can be easily seen that for $s>-1$ and $s\neq0$ the function $(\sgn s) f_s^w(x)$ \footnote{The function $\sgn s$ denotes the sign of $s$. It is defined to be $-1$ if $s<0$, $0$ if $s=0$, and $1$ if $s>0$.} is a convex and decreasing function. Hence, application of Proposition~\ref{prop:optimality}(A) reproves the well-known fact that the $N$-th roots of unity and their rotated copies are the only optimal $f_s^w$-energy configurations for $s$ in the range $(-1,0)\cup(0,\infty)$. (We remind the reader that, in contrast to the geodesic case, in the Euclidean case the $N$-th roots of unity are optimal for $s\geq-2$, $s\neq0$, and they are unique up to rotation for $s>-2$, see discussion in \cite{BrHaSa2009}.) The complete asymptotic expansion of $\mathcal{L}_s(N) = \mathcal{M}(\mathbb{S}^1,f_s^w; N)$ can be obtained from Theorem~\ref{thm:weighted.f} if $s$ is not an integer, from Theorem~\ref{thm:Laurent.series} if $s$ is a positive integer, and from Theorem~\ref{thm:analytic.f} if $s$ is a negative integer. (We leave the details to the reader.) For $s\in\mathbb{C}$ with $s\neq 0, 1, 3, 5, \dots$ and $q-2p < \re s < 2 + q$, the Euclidean Riesz $s$-energy for the $N$-th roots of unity is given by (cf. \cite[Theorem~1.1]{BrHaSa2009})
\begin{equation} \label{eq:cal.L.s.N}
\mathcal{L}_s(N) = V_s \, N^2 + \frac{2\zetafcn(s)}{(2\pi)^s} N^{1+s} + \sum_{n=1}^{q} \alpha_n(s) \frac{2\zetafcn(s-2n)}{(2\pi)^{s-2n}} N^{1+s-2n} + \mathcal{O}( N^{1-2p} ) + \mathcal{O}( N^{s-2p} ) 
\end{equation}
as $N \to \infty$, where (cf. \cite{BrHaSa2009})
\begin{align} \label{eq:V.s.alpha.n}
V_s &= \frac{2^{-s}\gammafcn((1-s)/2)}{\sqrt{\pi}\gammafcn(1-s/2)}, &\qquad \alpha_n(s) &= \frac{(-1)^n B_{2n}^{(s)}(s/2)}{(2n)!}, \quad n = 0, 1, 2, \dots.
\end{align}
Here, $B_n^{(\alpha)}(x)$ denotes the generalized Bernoulli polynomial, where $B_n(x) = B_n^{(1)}(x)$. Notice the absence of the term $\mathcal{B}_p(\Gamma, f_s^w; N)$, which follows from the fact that odd derivatives of $f_s^w(x)$ evaluated at $\pi$ assume the value $0$. (This can be seen, for example, from Fa\`{a} di Bruno's differentiation formula.) 

The entirety of positive odd integers $s$ constitutes the class of exceptional cases regarding the Euclidean Riesz $s$-energy of the $N$-th roots of unity. For such $s$ Theorem~\ref{thm:Laurent.series}(ii) provides the asymptotic expansion of $\mathcal{L}_s(N) = \mathcal{M}(\mathbb{S}^1,f_s^w;N)$, which features an $N^2 \log N$ term as leading term. That is, for $s=2L+1$, $L=0,1,2,\dots$, we have from Theorem~\ref{thm:Laurent.series}(ii) that (cf. \cite[Thm.~1.2]{BrHaSa2009})
\begin{equation} \label{eq:L.s.N}
\mathcal{L}_s(N) = \frac{\alpha_L(s)}{\pi} N^2 \log N + V_{f_s^w}(\mathbb{S}^1) N^2 + \sum_{\substack{m=0, \\ m \neq L}}^{p+L} \alpha_m(s) \frac{2\zetafcn(s-2m)}{\left( 2 \pi \right)^{s-2m}} N^{1+s-2m} + \mathcal{O}(N^{1-2p}),
\end{equation}
where the coefficients $\alpha_m(s)$ are given in \eqref{eq:V.s.alpha.n} and
\begin{equation*}
V_{f_s^w}(\mathbb{S}^1) = \frac{1}{\pi} \Bigg\{ \sum_{\substack{m=0, \\ m \neq L}}^\infty \alpha_m(s) \frac{\pi^{2m+1-s}}{2m+1-s} - \alpha_L(s) \left( \log 2 - \gamma \right) \Bigg\}.
\end{equation*}
We remark that in \cite[Thm.~1.2]{BrHaSa2009} we also give a computationally more accessible representation of $V_{f_s^w}(\mathbb{S}^1)$. The appearance of the $N^2 \log N$ terms can be understood on observing that the constant $V_s$ in \eqref{eq:cal.L.s.N} has its simple poles at positive odd integers $s$ and when using a limit process as $s\to K$ ($K$ a positive odd integer) in \eqref{eq:cal.L.s.N}, the simple pole at $s=K$ need to be compensated by the simple pole of the Riemann zeta function in the coefficient of an appropriate lower-order term. This interplay produces eventually the $N^2 \log N$ term. 


\section{The geodesic Riesz $s$-energy of equally spaced points on $\Gamma$}
\label{sec:geodesic.Riesz.s.energy}

Here, we state theorems concerning the geodesic Riesz $s$-energy of equally spaced points on $\Gamma$ that follow of the results from the preceding section together with asymptotic properties of generalized harmonic numbers. The proofs are given in Section~\ref{sec:proofs}. 

\begin{defn} \label{def:main}
The {\em discrete geodesic Riesz $s$-energy} of $N$ equally spaced points $\PT{z}_{1,N}, \dots, \PT{z}_{N,N}$ on $\Gamma$ is given by
\begin{equation*}
\mathcal{M}_s(\Gamma;N) \DEF \sum_{j\neq k} \left[ \gdist(\PT{z}_{j,N},\PT{z}_{k,N}) \right]^{-s} 
= N \sum_{j=1}^{N-1} \left[ \gdist(\PT{z}_{j,N},\PT{z}_{N,N}) \right]^{-s}, \qquad s \in \mathbb{C}.
\end{equation*}
The {\em discrete logarithmic geodesic energy} of $N$ equally spaced points $\PT{z}_{1,N}, \dots, \PT{z}_{N,N}$ on $\Gamma$ enters in a natural way by taking the limit
\begin{equation}
\mathcal{M}_{\mathrm{log}}(\Gamma; N) \DEF \lim_{s\to0} \frac{\mathcal{M}_s(\Gamma; N)-N(N-1)}{s} = \sum_{j\neq k} \log\frac{1}{\gdist(\PT{z}_{j,N},\PT{z}_{k,N})}. \label{M.0}
\end{equation}
\end{defn}

We are interested in the asymptotics of $\mathcal{M}_s(\Gamma; N)$ for large $N$ for all values of $s$ in the complex plane and we shall compare them with the related asymptotics for the Euclidean case given in our recent paper \cite{BrHaSa2009}. 
In the following we use the notation
\begin{align*}
\mathcal{I}_s^g[\mu] &\DEF \int \int \frac{\dd \mu(\PT{x}) \dd \mu(\PT{y})}{\left[ \gdist(\PT{x},\PT{y}) \right]^s}, &\qquad V_s^g(\Gamma) &\DEF \inf\{ \mathcal{I}_s^g[\mu] : \mu \in \mathfrak{M}(\Gamma) \}, \\
\mathcal{I}_{\mathrm{log}}^g[\mu] &\DEF \int \int \log \frac{1}{\gdist(\PT{x},\PT{y})} \dd \mu(\PT{x}) \dd \mu(\PT{y}), &\qquad V_{\mathrm{log}}^g(\Gamma) &\DEF \inf\{ \mathcal{I}_{\mathrm{log}}^g[\mu] : \mu \in \mathfrak{M}(\Gamma) \}.
\end{align*}

\subsection{The geodesic logarithmic energy}

\begin{thm} \label{thm.M.0} 
Let $q$ be a positive integer. For $N = 2 M + \kappa$, $\kappa = 0, 1$
\begin{equation*}
\mathcal{M}_{\mathrm{log}}(\Gamma; N) = V_{\mathrm{log}}^g( \Gamma) \, N^2 - N \log N + N \log \frac{\left| \Gamma \right|}{2\pi} - \sum_{n=1}^q \frac{B_{2n}(\kappa / 2)}{\left( 2 n - 1 \right) 2 n} 2^{2n} N^{2-2n} + \mathcal{O}_{q,\kappa}(N^{-2q})
\end{equation*}
as $N \to \infty$. Here, $V_{\mathrm{log}}^g( \Gamma) = 1 - \log ( | \Gamma | / 2 )$. 
\end{thm}

\begin{rmk}
The parity of $N$ affects the coefficients of the powers $N^{2-2m}$, $m \geq 1$.
The $N^2$-term vanishes for curves $\Gamma$ with $| \Gamma | = 2 e$ and the $N$-term vanishes when $| \Gamma | = 2 \pi$. 
By contrast, the Euclidean logarithmic energy of $N$ equally spaced points on the unit circle is given by (cf. \cite{BrHaSa2009})
\begin{equation*} 
\mathcal{L}_{\mathrm{log}}(N) = - N \log N.
\end{equation*}
\end{rmk}

\subsection{The geodesic Riesz $s$-energy}

The next result provides the complete asymptotic formula for all $s\neq1$. This exceptional case, in which a logarithmic term arises, is described in Theorem~\ref{thm:s.EQ.1}.


\begin{thm}[general case] \label{thm:main}
Let $q$ be a positive integer. Then for all $s \in \mathbb{C}$ with $s\neq 1$ and $\re s + 2q \geq 0$ there holds
\begin{equation}
\mathcal{M}_s(\Gamma; N) = V_s^g( \Gamma ) \, N^2 + \frac{2\zetafcn(s)}{\left| \Gamma \right|^s} N^{1+s} - \frac{1}{\left( \left| \Gamma \right| / 2 \right)^s} \sum_{n=1}^q \frac{B_{2n}(\kappa/2)}{(2n)!} \Pochhsymb{s}{2n-1} 2^{2n} N^{2-2n} + \mathcal{O}_{s,q,\kappa}(N^{-2q}) \label{gen:asympt.1}
\end{equation}
as $N\to\infty$, where $V_s^g( \Gamma ) = ( | \Gamma | / 2 )^{-s} / ( 1 - s )$ and $N = 2M + \kappa$, $\kappa = 0, 1$. 
\end{thm}

In \eqref{gen:asympt.1} the symbol $\Pochhsymb{s}{n}$ denotes the Pochhammer symbol defined as $\Pochhsymb{s}{0} = 1$ and $\Pochhsymb{s}{n+1} = ( n + s ) \Pochhsymb{s}{n}$ for integers $n\geq0$. 

\begin{rmk}
It is interesting to compare \eqref{gen:asympt.1} with \eqref{eq:cal.L.s.N}. 
%
%
%
It should be noted that in both the geodesic and the Euclidean case, the respective asymptotics have an $N^2$-term whose coefficient is the respective energy integral of the limit distribution (which is the normalized arc-length measure) or its appropriate analytic continuation, and an $N^{1+s}$-term with the coefficient $2 \zetafcn(s) / | \Gamma |^s$. Regarding the latter, it has been shown in \cite{MaMaRa2004} that for $s>1$ the dominant term of the asymptotics for the (Euclidean) Riesz $s$-energy of optimal energy $N$-point systems for any one-dimensional rectifiable curves in $\Rset^{p}$ is given by $2 \zetafcn(s) / | \Gamma |^s N^{1+s}$. 
Regarding the remaining terms of the asymptotics of $\mathcal{M}_s(\Gamma;N)$ and $\mathcal{L}_s(N)$ one sees that the exponents of the powers of $N$ do not depend on $s$ in the geodesic case but do depend on $s$ in the Euclidean case.
\end{rmk}

\begin{rmk}
In the general case $s\neq1$, the asymptotic series expansion \eqref{gen:asympt.1} is not convergent, except for $s=0,-1,-2, \dots$ when the infinite series reduces to a finite sum. The former follows from properties of the Bernoulli numbers and the latter from properties of the Pochhammer symbol $\Pochhsymb{a}{n}$. 
\end{rmk}

For a negative integer $s$ we have the following result.
\begin{prop} \label{prop:M.neg.p}
Let $p$ be a positive integer. Then
\begin{equation*}
\begin{split}
\mathcal{M}_{-p}(\Gamma; N) &= \frac{\left( \left| \Gamma \right| / 2 \right)^p}{p+1} N^2 + \frac{\left( \left| \Gamma \right| / 2 \right)^p}{p+1} \sum_{n=1}^{\lfloor p / 2 \rfloor} \binom{p+1}{2n} B_{2n}(\kappa/2) \, 2^{2n} N^{2-2n} \\
&\phantom{=\pm}+ \frac{2 \left| \Gamma \right|^p}{p+1} \left( B_{p+1}(\kappa/2) - B_{p+1} \right) N^{1-p}
\end{split} \label{M.s.spec.odd}
\end{equation*}
for $N = 2M + \kappa$, $\kappa = 0, 1$. The right-most term above vanishes for even $p$.
\end{prop}

%
\begin{rmk} The corresponding {\em Euclidean} Riesz $(-p)$-energy of $N$-th roots of unity reduces to 
\begin{equation*} 
\mathcal{L}_{-p}(N) = V_{-p} N^2 \quad \text{if $p=2,4,6,\dots$.}
\end{equation*}
\end{rmk}

\begin{rmk}
The quantity $\mathcal{M}_{-1}(\mathbb{S}; N)$ gives the maximum sum of {\em geodesic} distances on the unit circle. Corollary~\ref{prop:M.neg.p} yields
\begin{equation} \label{eq:M.neg.1}
\mathcal{M}_{-1}(\mathbb{S}; N) = \frac{\pi}{2} \left( N^2 - \kappa \right), \qquad \text{$N = 2 M + \kappa$, $\kappa = 0, 1$.}
\end{equation}
We remark that L. Fejes T{\'o}th~\cite{Fe1959} conjectured (and proved for $N \leq 6$) that the maximum sum of geodesic distances on the {\em unit sphere $\mathbb{S}^2$ in $\mathbb{R}^3$} is also given by the right-hand side in \eqref{eq:M.neg.1}. This conjecture was proved by Sperling~\cite{Sp1960} for even $N$ \footnote{Sperling mentions that his proof can be easily generalized to higher-dimensional spheres.} and by Larcher~\cite{La1962} for odd $N$.\footnote{Larcher also characterizes all optimal configurations.} An essential observation is that the sum of geodesic distances does not change if a given pair of antipodal points $(\PT{x}, \PT{x}^\prime)$ is rotated simultaneously, since $\gdist(\PT{x},\PT{y}) + \gdist(\PT{x}^\prime,\PT{y}) = \pi$ for every $\PT{y} \in \mathbb{S}^2$. 
\end{rmk}

In the exceptional case $s=1$ a logarithmic term appears. 

\begin{thm} \label{thm:s.EQ.1} 
Let $q \geq 1$ be an integer. For $N = 2 M + \kappa$, $\kappa = 0, 1$, 
\begin{equation} \label{M.s.EQ.1}
\begin{split}
\mathcal{M}_1(\Gamma; N) &= \frac{2}{\left| \Gamma \right|} N^2 \log N - \frac{\log2-\gamma}{\left| \Gamma \right| / 2} N^2 - \frac{2}{\left| \Gamma \right|} \sum_{n=1}^q \frac{B_{2n}(\kappa/2)}{2n} 2^{2n} N^{2-2n} \\
&\phantom{=\pm}- \theta_{q,N,\kappa} \frac{2}{\left| \Gamma \right|} \frac{B_{2q+2}(\kappa/2)}{2q+2} 2^{2q+2} N^{-2q},
\end{split}
\end{equation}
where $0 < \theta_{q,N,\kappa} \leq 1$ depends on $q$, $N$ and $\kappa$.
\end{thm}

\begin{rmk} \label{rmk:s.EQ.1}
A comparison of the asymptotics \eqref{M.s.EQ.1} and the corresponding result for the Euclidean Riesz $1$-energy of $N$-th roots of unity (cf.  \eqref{eq:L.s.N} and \cite[Thm.~1.2]{BrHaSa2009}), 
\begin{equation*}
\mathcal{L}_1(N) = \frac{1}{\pi} N^2 \log N + \frac{\gamma - \log ( \pi / 2 )}{\pi} N^2 + \sum_{n=1}^q \frac{(-1)^n B_{2n}(1/2)}{(2n)!} \frac{2\zetafcn(1-2n)}{\left(2\pi\right)^{1-2n}} N^{2-2n} + \mathcal{O}(N^{1-2q}),
\end{equation*}
shows that for $| \Gamma | = 2 \pi$ the dominant term is the same and the coefficients of all other powers of $N$ differ. The latter is obvious for the $N^2$-term, and for the $N^{2-2n}$-term, follows from the fact that the coefficient in \eqref{M.s.EQ.1} multiplied by $\pi$ is rational whereas the coefficient in the asymptotics for $\mathcal{L}_1(N)$ multiplied by $\pi$ is transcendental. Interestingly, except for $s=1$, there are no other exceptional cases with an $N^2 \log N$ term in the asymptotics of $\mathcal{M}_s(\Gamma; N)$, whereas in the asymptotics of $\mathcal{L}_s(N)$ there appears an $N^2 \log N$ term whenever $s$ is a positive integer, cf. \cite[Thm.~1.2]{BrHaSa2009}. 
\end{rmk}

%


\section{Proofs}
\label{sec:proofs}

\begin{proof}[Proof of Proposition~\ref{prop:optimality}]
{\bf Part (A).} The proof utilizes the ``winding number'' argument of L. Fejes T{\'o}th. 
The key idea is to regroup the terms in the sum in \eqref{G.f} with respect to its $m$ nearest neighbors ($m=1,\dots,N$) and then use convexity and Jensen's inequality. 

W.l.o.g. we assume that $\PT{w}_1, \dots, \PT{w}_N$ on $\Gamma$ are ordered such that $\PT{w}_k$ precedes $\PT{w}_{k+1}$ (denoted $\PT{w}_k \prec \PT{w}_{k+1} $). We identify $\PT{w}_{j+N}$ with $\PT{w}_j$ for $j=1,\dots,N-1$.
By convexity
\begin{equation} \label{eq:sum1}
\sum_{j=1}^N \sum_{\begin{subarray}{c} k = 1 \\ k \neq j \end{subarray}}^N f(\gdist(\PT{w}_j,\PT{w}_k)) = N \sum_{k=1}^{N-1} \Big[ \frac{1}{N} \sum_{j=1}^N f(\gdist(\PT{w}_j,\PT{w}_{j+k})) \Big] \geq N \sum_{k=1}^{N-1} f( \frac{1}{N} \sum_{j=1}^N \gdist( \PT{w}_j, \PT{w}_{j+k}) ).
\end{equation}
Let $\PT{z}_{1,N} \prec \dots \prec \PT{z}_{N,N}$ be $N$ equally spaced (with respect to the metric $\gdist$) points on $\Gamma$. Set $\PT{z}_{0,N} = \PT{z}_{N,N}$. Assuming further that this metric $\gdist$ also satisfies 
\begin{equation} \label{main.property}
\frac{1}{N} \sum_{j=1}^N \gdist( \PT{x}_j, \PT{x}_{j+k} ) \leq \gdist(\PT{z}_{0,N},\PT{z}_{k,N}), \qquad k = 1, \dots, N - 1,
\end{equation}
for every ordered $N$-point configuration $\PT{x}_1 \prec \dots \prec \PT{x}_N$ with $\PT{x}_j = \PT{x}_{j+N}$, it follows that
\begin{equation*}
G_f(\PT{w}_1, \dots, \PT{w}_N) \geq N \sum_{k=1}^{N-1} f( \gdist(\PT{z}_{0,N},\PT{z}_{k,N})) \FED \mathcal{M}_f(\Gamma; N) = G_f(\PT{z}_{1,N},\dots, \PT{z}_{N,N}).
\end{equation*}
It remains to show that the geodesic distance satisfies \eqref{main.property}. From
\begin{equation*}
\gdist(\PT{x}_j,\PT{x}_k) = \min \left\{ \ell( \PT{x}_j, \PT{x}_k ), | \Gamma | - \ell( \PT{x}_j, \PT{x}_k ) \right\} \qquad \text{if $0 \leq k - j < N$}
\end{equation*}
and additivity of the distance function $\ell( \PT{\cdot}, \PT{\cdot} )$ it follows that
\begin{equation*}
\sum_{j=1}^N \gdist( \PT{x}_j, \PT{x}_{j+k} ) \leq 
\begin{cases}
\displaystyle \sum_{j=1}^N \ell( \PT{x}_j, \PT{x}_{j+k} ) = \sum_{j=1}^N \sum_{n=1}^k \ell( \PT{x}_{j+n-1}, \PT{x}_{j+n} ) = \left| \Gamma \right| k, \\
\displaystyle \sum_{j=1}^N \left( | \Gamma | - \ell( \PT{x}_j, \PT{x}_{j+k} ) \right) = \left| \Gamma \right| \left( N - k \right)
\end{cases}
\end{equation*}
and therefore
\begin{equation*}
\frac{1}{N} \sum_{j=1}^N \gdist( \PT{x}_j, \PT{x}_{j+k} ) \leq \min\{ \left| \Gamma \right| k / N,  \left| \Gamma \right| \left( N - k \right) / N \} = \gdist(\PT{z}_{0,N},\PT{z}_{k,N}).
\end{equation*}

In the case of a strictly convex function $f$ we have equality in \eqref{eq:sum1} if and only if the points are equally spaced. This shows uniqueness (up to translation along the simple closed curve $\Gamma$) of equally spaced points.

{\bf Part (B).} Given $N = 2 M + \kappa$ ($\kappa = 0, 1$) let $\omega_N$ denote the antipodal set with $M + \kappa$ points placed at the North Pole and $M$ points at the South Pole of $\Gamma$, where both Poles can be any two points on $\Gamma$ with geodesic distance $|\Gamma|/2$.
Thus, the geodesic distance between two points in $\omega_N$ is either $0$ or $|\Gamma|/2$. Hence
\begin{equation} \label{G.f.gen}
G_f(\omega_N) = 2 M \left( M + \kappa \right) f( \left| \Gamma \right| / 2 ) = \frac{1}{2} f( \left| \Gamma \right| / 2 ) \left( N^2 - \kappa \right).
\end{equation}

Since adding a constant to $G_f$ does not change the positions of optimal $f$-energy points, we may assume w.l.o.g. that $f(0)=0$. In fact, we will prove the equivalent assertion that if $f$ is a non-constant convex and increasing function with $f(0) = 0$, then the functional $G_f$ has a maximum at $\omega_N$, which is unique (up to translation along $\Gamma$) if $f$ is strictly increasing. (Note that by these assumptions $f(x) \geq 0$.) Indeed, any $N$-point system $X_N$ of points $\PT{x}_1, \dots, \PT{x}_N$ from $\Gamma$ satisfies
\begin{align*}
G_f(X_N) 
&= f( \left| \Gamma \right| / 2 ) \sum_{j \neq k} \frac{f( \gdist(\PT{x}_j, \PT{x}_k) )}{f( \left| \Gamma \right| / 2 )} \leq  f( \left| \Gamma \right| / 2 ) \sum_{j \neq k} \frac{\gdist(\PT{x}_j, \PT{x}_k)}{\left| \Gamma \right| / 2}= f( \left| \Gamma \right| / 2 ) \frac{G_{\mathrm{id}}(X_N)}{\left| \Gamma \right| / 2} \\
&\leq f( \left| \Gamma \right| / 2 ) \frac{G_{\mathrm{id}}(\omega_N)}{\left| \Gamma \right| / 2} = \frac{1}{2} f( \left| \Gamma \right| / 2 ) \left( N^2 - \kappa \right),
\end{align*}
where we used that antipodal configurations are optimal for the ``sum of distance function'' ($f$ is the identity function $\mathrm{id}$) and relation \eqref{G.f.gen} with $f \equiv \mathrm{id}$. Note that the first inequality is strict if there is at least one pair $(j,k)$ such that $0 < \gdist(\PT{x}_j, \PT{x}_k) < | \Gamma | / 2$. On the other hand, if $X_N = \omega_N$, then equality holds everywhere. 
\end{proof}

\begin{proof}[Proof of Proposition~\ref{prop:leading.term}]
For Lebesgue integrable functions $f$ the minimum geodesic $f$-energy $V_f^g(\Gamma)$ is finite, since $\mathcal{I}_f^g[\sigma_\Gamma] = \int f( \gdist(\PT{x}, \PT{y}) ) \dd \sigma_\Gamma(\PT{x}) = ( 2 / | \Gamma | ) \int_0^{| \Gamma | / 2} f(\ell) \dd \ell \neq \infty$ ($\PT{y} \in \Gamma$ arbitrary). 
Moreover, for lower semicontinuous functions $f$, a standard argument (see \cite{La1972}) shows that the sequence $\{G_f(\omega_N^{(f)}) / [ N ( N - 1 ) ] \}_{N\geq2}$
is monotonically increasing. Since $f$ is Lebesgue integrable, this sequence is bounded from above by $\mathcal{I}_f^g[\sigma_\Gamma]$; thus, the limit $\lim_{N\to\infty} G_f(\omega_N^{(f)}) / N^2$ exists in this case. If $f$ also satisfies the hypotheses of Proposition~\ref{prop:optimality}(A), then $\lim_{N\to\infty} G_f(\omega_N^{(f)}) / N^2 = \mathcal{I}_f^g[\sigma_\Gamma]$. (By a standard argument, one constructs a family of continuous functions $F_\eps(x)$ with $F_\eps(x) = f(x)$ outside of $\eps$-neighborhoods at points of discontinuity of $f$, $f(x) \geq F_\eps(x)$ everywhere and $\lim_{\eps\to0} F_\eps(x) = f(x)$ wherever $f$ is continuous at $x$. Then the lower bound follows from weak-star convergence of $\nu[\omega_N^{(f)}]$ as $N \to \infty$ and, subsequently, letting $\eps \to 0$.)
\end{proof}

We next present some auxiliary results that are needed to prove the main Theorems~\ref{thm:general.f.general.case} and ~\ref{thm:general.f.exceptional.case}. We begin with the following generalized Euler-MacLaurin summation formula.

\begin{prop} \label{prop:Euler-MacLaurin.Summation}
Let $\omega = 0$ or $\omega = 1/2$. Let $M \geq 2$. Then for any function $h$ with continuous derivative of order $2p+1$ on the interval $[1 - \omega, M + \omega]$ we have
\begin{equation*}
\begin{split}
\sum_{k = 1}^M h(k) 
&= \int_a^b h(x) \dd x + \left( 1 / 2 - \omega \right) \left\{ h(a) + h(b) \right\} + \sum_{k = 1}^p \frac{B_{2k}(\omega)}{(2k)!} \left\{ h^{(2k-1)}(b) - h^{(2k-1)}(a) \right\} \\
&\phantom{=}+ \frac{1}{(2p+1)!} \int_a^b C_{2p+1}(x) h^{(2p+1)}(x) \dd x, \qquad a = 1 - \omega, b = M + \omega,
\end{split}
\end{equation*}
where $C_{k}(x)$ is the periodized Bernoulli polynomial $B_{k}(x-\lfloor x\rfloor)$.
\end{prop}
\begin{proof}
For $\omega = 0$, the above formula is the classical Euler-MacLaurin summation formula (cf., for example, \cite{Ap1999}). For $\omega = 1 / 2$, iterated application of integration by parts yields the desired result. 
\end{proof}

Let $f$ have a continuous derivative of order $2p+1$ on the interval $(0 , | \Gamma | / 2]$. Then applying Proposition~\ref{prop:Euler-MacLaurin.Summation} with $h(x) = f( x | \Gamma | / N )$ and $\omega = \kappa / 2$, where $N = 2M+\kappa\geq2$, $\kappa=0,1$, we obtain
\begin{align*}
\mathcal{M}(\Gamma, f; N) 
&= 2 N \sum_{n = 1}^{\lfloor N / 2 \rfloor} f( n \left| \Gamma \right| / N ) - \left( 1 - \kappa \right) f( \left| \Gamma \right| / 2 ) N = 2 N \int_{1-\omega}^{N/2} f( x | \Gamma | / N ) \dd x \\
&\phantom{=}+ 2 \left( \frac{1}{2} - \omega \right) N \left\{ f( (1-\omega) | \Gamma | / N ) + f( | \Gamma | / 2 ) \right\} + 2 N \sum_{k = 1}^p \frac{B_{2k}(\omega)}{(2k)!} \left. \left\{f( x | \Gamma | / N )\right\}^{(2k-1)} \right|_{1-\omega}^{N/2} \\
&\phantom{=}+ \frac{2 N}{(2p+1)!} \int_{1-\omega}^{N/2} C_{2p+1}(x) \left\{f( x | \Gamma | / N )\right\}^{(2p+1)}(x) \dd x - 2 \left( \frac{1}{2} - \omega \right) f( \left| \Gamma \right| / 2 ) N.
\end{align*}
Regrouping the terms in the last relation and using the fact that $B_{2k+1} = B_{2k+1}(1/2) =0$ for $k=1,2,3,\dots$ and $B_1(\omega)=\omega-1/2$, we derive the exact representation
\begin{equation}
\mathcal{M}(\Gamma, f; N) = N^2 \, \frac{2}{\left| \Gamma \right|} \int_{\left( 1 - \omega \right) \left| \Gamma \right| / N}^{\left| \Gamma \right| / 2} f(y) \dd y - \mathcal{A}_p(\Gamma, f; N) + \mathcal{B}_p(\Gamma, f; N) + \mathcal{R}_p(\Gamma, f; N) \label{eq:term.general}
\end{equation}
valid for every integer $N \geq 2$, where
\begin{subequations}
\begin{align}
\mathcal{A}_p(\Gamma, f; N) &\DEF - 2 B_1(\omega) f( (1-\omega) | \Gamma | / N ) - 2 N \sum_{k = 1}^p \frac{B_{2k}(\omega)}{(2k)!} \left. \left\{f( x | \Gamma | / N )\right\}^{(2k-1)} \right|_{1-\omega} \notag \\
&= \frac{2}{\left| \Gamma \right|} N^2 \sum_{r = 1}^{2p} \frac{B_{r}(\omega)}{r!} \left( \left| \Gamma \right| / N \right)^{r} f^{(r-1)}( \left( 1 - \omega \right) \left| \Gamma \right| / N ), \label{eq:A.p} \\
\mathcal{B}_p(\Gamma, f; N) &\DEF \frac{2}{\left| \Gamma \right|} N^2 \sum_{k = 1}^{p} \frac{B_{2k}(\omega)}{(2k)!} \left( \left| \Gamma \right| / N \right)^{2k} f^{(2k-1)}( \left| \Gamma \right| / 2 ), \label{eq:B.p...} \\
\mathcal{R}_p(\Gamma, f; N) &\DEF 2 N \frac{\left( \left| \Gamma \right| / N \right)^{2p+1}}{(2p+1)!} \int_{1-\omega}^{N / 2} C_{2p+1}( x ) f^{(2p+1)}( x \left| \Gamma \right| / N) \dd x. \label{eq:R.p}
\end{align}
\end{subequations}

If $f$ is admissible in the sense of Definition~\ref{def:admissible}, then by linearity
\begin{equation*}
\mathcal{M}(\Gamma, f; N) = \mathcal{M}(\Gamma, S_q; N) + \mathcal{M}(\Gamma, f - S_q; N),
\end{equation*}
where the term $\mathcal{M}(\Gamma, S_q; N)$ contains the asymptotic expansion of $\mathcal{M}(\Gamma, f; N)$ and the term $\mathcal{M}(\Gamma, f - S_q; N)$ is part of the remainder term. 
The next lemma provides estimates for the contributions to the remainder term in the asymptotic expansion of $\mathcal{M}(\Gamma, f; N)$ as $N\to\infty$.

\begin{lem} \label{lem:estimates}
Let $f$ be admissible in the sense of Definition~\ref{def:admissible}. Then as $N\to\infty$:
\begin{align*}
N^2 \frac{2}{\left| \Gamma \right|} \int_0^{\left( 1 - \omega \right) \left| \Gamma \right| / N} (f - S_q)(y) \dd y &= \mathcal{O}( N^{1-\delta+s_q} ), \\
\mathcal{A}_p(\Gamma, f - S_q; N) &= \mathcal{O}( N^{1-\delta+s_q}), \\
\mathcal{R}_p(\Gamma, f - S_{q}; N) &= 
\begin{cases}
\displaystyle \mathcal{O}( N^{1-2p} ) & \text{if $2p \neq \delta - \re s_q$,} \\[1em]
\displaystyle \mathcal{O}( N^{1-2p} \log N ) & \text{if $2p = \delta - \re s_q$.}
\end{cases}
\end{align*}
The $\mathcal{O}$-term depends on $|\Gamma|$, $p$, $s_q$, and $f$. 
\end{lem}

\begin{proof}
The first relation follows directly from Definition~\ref{def:admissible}(ii.a). The second estimate follows from Definition~\ref{def:admissible}(ii.b) and \eqref{eq:A.p}; that is for some positive constant $C$
\begin{align*}
\left| \mathcal{A}_p(\Gamma, f - S_q; N) \right| 
&\leq \frac{2}{\left| \Gamma \right|} N^2 \sum_{r = 1}^{2p} \frac{| B_{r}(\omega) |}{r!} \left( \left| \Gamma \right| / N \right)^{r} \left| (f-S_q)^{(r-1)}( \left( 1 - \omega \right) \left| \Gamma \right| / N ) \right| \\
&\leq C \frac{2}{\left| \Gamma \right|} N^2 \sum_{r = 1}^{2p} \frac{| B_{r}(\omega) |}{r!} \left( \left| \Gamma \right| / N \right)^{r} \left( 1 - \omega \right)^{\delta-\re s_q-r+1} \left( \left| \Gamma \right| / N \right)^{r + \delta - \re s_q - r + 1}. 
\end{align*}
The last estimate follows from Definition~\ref{def:admissible}(ii.b), \eqref{eq:R.p} and the fact that 
\begin{equation} \label{eq:period.Bernoulli.estimate}
\left| C_{2p+1}(x) \right| \leq \left( 2 p + 1 \right) \left| B_{2p} \right| \qquad \text{for all real $x$ and all $p=1,2, \dots$;}
\end{equation}
that is for some positive constant $C$
\begin{align*}
\left| \mathcal{R}_p(\Gamma, f-S_Q; N) \right| 
&\leq  2 N \frac{\left( \left| \Gamma \right| / N \right)^{2p+1}}{(2p+1)!} \int_{1-\omega}^{N / 2} \left| C_{2p+1}( x )\right| \left| (f-S_q)^{(2p+1)}( x \left| \Gamma \right| / N) \right| \dd x \\
&\leq 2 C N \frac{B_{2p}}{(2p)!} \left( \left| \Gamma \right| / N \right)^{\delta-\re s_q} \int_{1-\omega}^{N / 2} x^{\delta-1-2p-\re s_q} \dd x.
\end{align*}
\end{proof}

Other functions arising in the asymptotics of $\mathcal{M}(\Gamma, f; N)$ are defined next.
\begin{defn} \label{def:zeta.psi}
Let $\omega = 0, 1/2$ and $p$ be a positive integer. For $s \in \mathbb{C}$ with $s \neq 1$
\begin{equation*}
\begin{split}
\zetafcn_p(\omega, y; s) 
&\DEF \frac{1}{s-1} \sum_{r=0}^{2p} \frac{B_r(\omega)}{r!} ( -1 )^r \Pochhsymb{s-1}{r} \left( 1 - \omega \right)^{1-s-r} - \frac{\Pochhsymb{s}{2p+1}}{\left( 2p + 1\right)!} \int_{1-\omega}^y C_{2p+1}(x) x^{-s-1-2p} \dd x, 
\end{split}
\end{equation*}
which we call {\em incomplete zeta function} and
\begin{equation*}
\begin{split}
\Psi_p(\omega, y) 
&\DEF  - \log( 1 - \omega ) + \sum_{r=1}^{2p} \frac{B_r(\omega)}{r} ( -1 )^r \left( 1 - \omega \right)^{-r} - \int_{1-\omega}^y C_{2p+1}(x) x^{-2-2p} \dd x.
\end{split}
\end{equation*}
\end{defn}

\begin{prop} \label{prop:aux.results}
Let $\omega = 0, 1/2$. Then 
\begin{align}
\Psi_p(\omega, y) &= \lim_{s \to 1} ( \zetafcn_p(\omega, y; s) - 1 / (s-1) ), \notag \\
\zetafcn_p(\omega, y; -n) &= - \frac{B_{n+1}}{n+1} = \zetafcn(-n), \qquad n = 0, 1, \dots, 2p, \notag \\
\zetafcn_p(\omega, y; s) - \zetafcn(s) &= \frac{\Pochhsymb{s}{2p+1}}{\left( 2p + 1\right)!} \int_y^\infty C_{2p+1}(x) x^{-s-1-2p} \dd x, \qquad \re s + 2p > 0, \notag \\
\zetafcn(s) &= \lim_{y\to\infty} \zetafcn_p(\omega, y; s), \qquad \re s + 2p > 0, \notag \\
\Psi_p(\omega, y) - \gamma &= \int_y^\infty C_{2p+1}(x) x^{-2-2p} \dd x, \notag \\
\gamma &= \lim_{y\to\infty} \Psi_p(\omega, y). \notag
\end{align}
\end{prop}

\begin{proof}
The second relation follows from \cite[Eq.~2.8(13)]{Lu1969I}, $B_{2k+1}(\omega) = 0$ for $\omega=0,1/2$ and $k\geq1$ and \cite[Eq.~23.2.15]{AbSt1992}. The representations and therefore the limit relations for $\zetafcn(s)$ and $\gamma$ follow from Proposition~\ref{prop:Euler-MacLaurin.Summation}.
\end{proof}

\begin{proof}[Proof of Theorem~\ref{thm:general.f.general.case}]
Let $f$ be admissible in the sense of Definition~\ref{def:admissible}. 
In the representation \eqref{eq:term.general} we can write the integral as follows:
\begin{align*}
\frac{2}{\left| \Gamma \right|} \int_{a}^{\left| \Gamma \right| / 2} f(y) \dd y 
&= \frac{2}{\left| \Gamma \right|} \int_{a}^{\left| \Gamma \right| / 2} S_q(x) \dd x + \frac{2}{\left| \Gamma \right|} \int_{a}^{\left| \Gamma \right| / 2} ( f - S_q )(x) \dd x \notag \\
&= \frac{2}{\left| \Gamma \right|} \sum_{n=0}^q a_n \int_a^{\left| \Gamma \right| / 2} x^{-s_n} \dd x + \frac{2}{\left| \Gamma \right|} \int_0^{\left| \Gamma \right| / 2} ( f - S_q )(x) \dd x - \frac{2}{\left| \Gamma \right|} \int_0^a ( f - S_q )(x) \dd x \\
&= V_f(\Gamma) - \frac{2}{\left| \Gamma \right|} \sum_{n=0}^q a_n \frac{a^{1-s_n}}{1-s_n} - \frac{2}{\left| \Gamma \right|} \int_0^a ( f - S_q )(x) \dd x, \qquad a \DEF ( 1 - \omega ) | \Gamma | / N.
\end{align*}
Defining 
\begin{equation*}
\tilde{\mathfrak{R}}_p(f-S_q; N) \DEF - \frac{2}{\left| \Gamma \right|} N^2 \int_0^{\left( 1 - \omega \right) \left| \Gamma \right| / N} ( f - S_q )(x) \dd x - \mathcal{A}_p(\Gamma, f - S_q; N) + \mathcal{R}_p(\Gamma, f - S_{q}; N),
\end{equation*}
formula \eqref{eq:term.general} becomes (in condensed notation)
\begin{align*}
\mathcal{M}(f; N) 
&= V_f \, N^2 - \frac{2}{\left| \Gamma \right|} N^2 \sum_{n=0}^q a_n \frac{a^{1-s_n}}{1-s_n} - \mathcal{A}_p(S_q; N) + \mathcal{B}_p(f; N) + \mathcal{R}_p(S_q; N) + \tilde{\mathfrak{R}}_p(f-S_q; N) \\
&= V_f \, N^2 + \sum_{n=0}^q a_n \left\{ \frac{2}{\left| \Gamma \right|} N^2 \frac{a^{1-s_n}}{s_n-1} - \mathcal{A}_p(x^{-s_n}; N) + \mathcal{R}_p(x^{-s_n}; N) \right\} + \mathcal{B}_p(f; N) + \tilde{\mathfrak{R}}_p(f-S_q; N).
\end{align*}
Furthermore, using \eqref{eq:A.p}, \eqref{eq:R.p} and Definition~\ref{def:zeta.psi}, we can write the expression in curly brackets above as follows:
\begin{align*}
\frac{2}{\left| \Gamma \right|} N^2 & \frac{a^{1-s_n}}{s_n-1} - \mathcal{A}_p(x^{-s_n}; N) + \mathcal{R}_p(x^{-s_n}; N) = \frac{2}{\left| \Gamma \right|} N^2 \frac{a^{1-s_n}}{s_n-1} - \frac{2}{\left| \Gamma \right|} N^2 \sum_{r = 1}^{2p} \frac{B_{r}(\omega)}{r!} \left( \left| \Gamma \right| / N \right)^{r} \left. \left\{ t^{-s_n} \right\}^{(r-1)} \right|_{t=a} \\
&\phantom{=}+ 2 N \frac{\left( \left| \Gamma \right| / N \right)^{2p+1}}{(2p+1)!} \int_{1-\omega}^{N / 2} C_{2p+1}( x ) \left. \left\{ t^{-s_n} \right\}^{(2p+1)} \right|_{t=x \left| \Gamma \right| / N} \dd x \\
&= \frac{2}{\left| \Gamma \right|} N^2 \left( \left| \Gamma \right| / N \right)^{1-s_n} \Bigg\{ \frac{\left( 1 - \omega \right)^{1-s_n}}{s_n-1} + \sum_{r=1}^{2p} \frac{B_r(\omega)}{r!} (-1)^r \Pochhsymb{s_n}{r-1} \left( 1 - \omega \right)^{1-s_n-r} \\
&\phantom{=}- \frac{\Pochhsymb{s_n}{2p+1}}{(2p+1)!} \int_{1-\omega}^{N/2} C_{2p+1}(x) \, x^{-s_n-1-2p} \dd x \Bigg\} = \frac{2}{\left| \Gamma \right|} N^2 \left( \left| \Gamma \right| / N \right)^{1-s_n} \zetafcn_p(\omega,N/2;s_n).
\end{align*}
Hence, we arrive at the formula
\begin{equation*}
\mathcal{M}(f; N) = V_f \, N^2 + \sum_{n=0}^q a_n \frac{2 \zetafcn_p(\omega,N/2;s_n)}{\left| \Gamma \right|^{s_n}} N^{1+s_n} + \mathcal{B}_p(f; N) + \tilde{\mathfrak{R}}_p(f-S_q; N).
\end{equation*}

For $\mathfrak{R}_p(\Gamma, f; N)$ defined by \eqref{eq:general.case.asymptotics} we have
\begin{equation} \label{eq:main.proof.aux2}
\mathfrak{R}_p(\Gamma, f; N) = \sum_{n=0}^q a_n \frac{2 \zetafcn_p(\kappa/2,N/2;s_n) - 2\zetafcn(s_n)}{\left| \Gamma \right|^{s_n}} N^{1+s_n} + \tilde{\mathfrak{R}}_p(f-S_q; N).
\end{equation}
Furthermore, it follows from Lemma~\ref{lem:estimates} that $\tilde{\mathfrak{R}}_p(f-S_q; N) = \mathcal{O}( N^{1-\delta+s_q}) + \mathcal{O}( N^{1-2p} )$ if $2p \neq \delta - \re s_q$ and $\tilde{\mathfrak{R}}_p(f-S_q; N) = \mathcal{O}( N^{1-\delta+s_q}) + \mathcal{O}( N^{1-2p} \log N )$ if $2p = \delta - \re s_q$. Finally, using \eqref{eq:period.Bernoulli.estimate} and Proposition~\ref{prop:aux.results} we obtain the estimate
\begin{equation*}
\left| \sum_{n=0}^q a_n \frac{\zetafcn_p(\kappa/2,N/2,s_n) - \zetafcn(s_n)}{\left| \Gamma \right|^{s_n}} N^{1+s_n} \right| \leq 2 \left( N / 2 \right)^{1-2p} \sum_{n=0}^q \left| a_n \frac{B_{2p}}{(2p)!} \Pochhsymb{s_n}{2p} \frac{2p+s_n}{2p+\re s_n} \right| \left( \left| \Gamma \right| / 2 \right)^{-\re s_n}.
\end{equation*}
Note that, whenever $s_n=-k$ for some $k=0,1,\dots,2p$, then the corresponding terms on both sides of the estimate above are not present. Also, from Definition~\ref{def:admissible} it follows that $2p+\re s_n>0$ for $n=0,\dots,q-1$ and that either $\re s_q + 2p > 0$ or $s_q=-2p$. In either case the sum on the left-hand side above is of order $\mathcal{O}(N^{1-2p})$. Hence, we have from \eqref{eq:main.proof.aux2} that $\mathfrak{R}_p(\Gamma, f; N) = \mathcal{O}( N^{1-\delta+s_q}) + \mathcal{O}( N^{1-2p} )$ if $2p \neq \delta - \re s_q$ and $\mathfrak{R}_p(\Gamma, f; N) = \mathcal{O}( N^{1-\delta+s_q}) + \mathcal{O}( N^{1-2p} \log N )$ if $2p = \delta - \re s_q$.
\end{proof}

\begin{proof}[Proof of Theorem~\ref{thm:general.f.exceptional.case}]
Proceeding as in the proof of Theorem~\ref{thm:general.f.general.case} the remainder term now takes the form
\begin{align*}
\mathfrak{R}_p(\Gamma, f; N) 
&= \frac{2}{\left| \Gamma \right|} N^2 a_{q^\prime} \left( \Psi_p(\kappa/2,N/2) - \gamma \right) + \sum_{\substack{n=0,\\ n\neq q^\prime}}^q a_n \frac{2 \zetafcn_p(\kappa/2,N/2,s_n) - 2\zetafcn(s_n)}{\left| \Gamma \right|^{s_n}} N^{1+s_n} \\
&\phantom{=}- N^2 \frac{2}{\left| \Gamma \right|} \int_0^{\left( 1 - \omega \right) \left| \Gamma \right| / N} (f-S_q)(y) \dd y - \mathcal{A}_p(\Gamma, f - S_q; N) + \mathcal{R}_p(\Gamma, f - S_{q}; N).
\end{align*}
Using Lemma~\ref{lem:estimates}, Proposition~\ref{prop:aux.results}, and the inequality
\begin{equation*}
\left| \frac{2}{\left| \Gamma \right|} N^2 a_{q^\prime} \left( \Psi_p(\kappa/2,N/2) - \gamma \right) \right| \leq 4 \frac{2}{\left| \Gamma \right|} \left| a_{q^\prime} B_{2p} \right| \left( N / 2 \right)^{1-2p},
\end{equation*}
we get the estimate $\mathfrak{R}_p(\Gamma, f; N) = \mathcal{O}( N^{1-2p} ) + \mathcal{O}( N^{1-\delta+s_q} )$ if $2p \neq \delta - \re s_q$ and $\mathfrak{R}_p(\Gamma, f; N) = \mathcal{O}( N^{1-2p} \log N)$ if $2p = \delta - \re s_q$.
\end{proof}

Next, we prove the results related to particular types of kernel functions. 

\begin{proof}[Proof of Theorem~\ref{thm:completely.monotonic}]
The Laplace transform $f(x) \DEF \int_0^\infty e^{-x t} \dd \mu(t)$ of a signed measure $\mu$ on $[0,\infty)$ satisfying $\int_0^\infty t^m \dd |\mu|(t) < \infty $ for every $m=0, 1, 2, \dots$ has derivatives of all orders on $(0,\infty)$. For $q$ a positive integer let $S_q(x)$ be defined by $S_q(x) \DEF \sum_{n=0}^q \frac{\mu_n}{n!} (-x)^n$. For every $0\leq m \leq q$ we can write 
\begin{equation*}
f^{(m)}(x) = (-1)^m \int_0^\infty e^{-x t} t^m \dd \mu(t) = (-1)^m \sum_{n=m}^q \frac{\mu_n}{(n-m)!} (- x)^{n-m} + (f-S_q)^{(m)}(x), \qquad x > 0, 
\end{equation*}
where, using a finite section of the Taylor series expansion of $h(x) = e^{- x t}$ with integral remainder term, we have that
\begin{align*}
(f-S_q)^{(m)}(x) &= f^{(m)}(x)-S_{q}^{(m)}(x) 
= (-1)^m \int_0^\infty \left\{ e^{-x t} - \sum_{n=0}^{q-m} \frac{(-x t)^n}{n!} \right\} t^m \dd \mu(t) \\
&= \frac{(-1)^{q+1}}{(q-m)!} \int_0^\infty \left\{ \int_0^x e^{-u t} \left( x - u \right)^{q-m} \dd u \right\} t^{q+1} \dd \mu(t),  \qquad x > 0.
\end{align*}
For $x>0$ we have the following bound:
\begin{equation*}
\left| (f-S_q)^{(m)}(x) \right| \leq \frac{x^{q+1-m}}{(q+1-m)!} \int_0^\infty t^{q+1} \dd |\mu|(t), \qquad m = 0, 1, \dots, q.
\end{equation*}
Since $S_q^{(q+1)}(x) = 0$ for all $x$, it is immediate that the last estimate also holds for $m=q+1$.
It follows that $f$ is admissible in the sense of Definition~\ref{def:admissible} with $q=2p$, $\delta=1$. 
The result follows from Theorem~\ref{thm:general.f.general.case}, after observing that
\begin{equation*}
V_f(\Gamma) = \frac{2}{\left| \Gamma \right|} \int_0^{\left| \Gamma \right|/2} f(x) \dd x = \frac{2}{\left| \Gamma \right|} \int_0^{\left| \Gamma \right|/2} \int_0^\infty e^{-x t} \dd \mu(t) \dd x.
\end{equation*}
\end{proof}

%

In the case that $f$ is a completely monotonic function on $(0,\infty)$ (that is, $\mu$ is a positive measure), it is possible to improve the estimate for $\mathcal{R}_p(\Gamma, f; N)$ in \eqref{eq:R.p}.

\begin{proof}[Proof of Theorem~\ref{thm:analytic.f}]
Let $f$ be analytic in a disc with radius $| \Gamma | / 2 + \eps$ ($\eps > 0$) centered at the origin. Then $f(z) = \sum_{n=0}^\infty a_n z^n$ for $|z| < | \Gamma | / 2 + \eps$ and $f$ is admissible in the sense of Definition~\ref{def:admissible} for any positive integers $p$ and $q=2p$, where $S_{2p}(z) = \sum_{n=0}^{2p} a_n z^n$ and $\delta=1$. The asymptotic expansion follows from Theorem~\ref{thm:general.f.general.case} on observing that with $s_n=-n$ ($n=0,\dots,2p$), one has
\begin{equation*}
V_f(\Gamma) = \frac{2}{\left| \Gamma \right|} \sum_{n=0}^{2p} a_n \int_0^{\left| \Gamma \right|/2} x^n \dd x + \frac{2}{\left| \Gamma \right|} \int_0^{\left| \Gamma \right|/2} \left( f - S_{2p} \right)(x) \dd x = \frac{2}{\left| \Gamma \right|} \int_0^{\left| \Gamma \right|/2} f(x) \dd x.
\end{equation*}
Moreover, since $s_q=-2p$ and $\delta=1$, it follows that $\mathfrak{R}_p(\Gamma, f; N) = \mathcal{O}_{p,|\Gamma|,f}(N^{1-2p})$ as $N\to\infty$.
\end{proof}

\begin{proof}[Proof of Theorem~\ref{thm:Laurent.series}]
Suppose $f$ has a pole of integer order $K\geq1$ at zero and is analytic in the annulus $0 < |z| < | \Gamma | / 2 + \eps$ ($\eps>0$) with series expansion $f(z) = \sum_{n=-K}^{\infty} a_n z^n$. Then $f$ is admissible in the sense of Definition~\ref{def:admissible} for any positive integers $p$ and $q=2p$ with $S_{2p}(z) = \sum_{n=-K}^{2p} a_n z^n$ and $\delta=1$. In the {case {\rm (i)}} Theorem~\ref{thm:general.f.general.case} is applied and in the {case {\rm (ii)}} Theorem~\ref{thm:general.f.exceptional.case} is applied. The expressions for $V_f(\Gamma)$ follow from termwise integration in \eqref{eq:general.case.V.f} and \eqref{eq:except.case.V.f}. Since $1-\delta+s_q=-2p$, the remainder terms are $\mathfrak{R}_p(\Gamma, f; N) = \mathcal{O}_{p,|\Gamma|,f}(N^{1-2p})$ as $N\to\infty$. 
\end{proof}

\begin{proof}[Proof of Examples~\ref{eg:ess.sing.1} and \ref{eg:ess.sing.2}]
If $f$ has an essential singularity at $0$ and is analytic in the annulus $0 < |z| < |\Gamma|/2+\eps$ ($\eps>0$), then for positive integers $p$
\begin{equation*}
f(z) = S_{2p}(z) + F_{2p}(z), \qquad S_{2p}(z) \DEF \sum_{n=-\infty}^{2p} a_n z^{n}, \quad F_{2p}(z) \DEF \sum_{n=2p+1}^\infty a_n z^{n} = \mathcal{O}(z^{2p+1}) \ \text{as $z\to0$.}
\end{equation*}
Clearly, the function $f(z)$ satisfies item (i) of Definition~\ref{def:admissible} and both functions $f(z)$ and $S_{2p}(z)$ satisfy an extended version of item (ii) of Definition~\ref{def:admissible} suitable for an infinite series $S_{2p}(z)$. Since termwise integration and differentiation of $S_{2p}(z)$ are justified by the theory for Laurent series, Theorems~\ref{thm:general.f.general.case} and \ref{thm:general.f.exceptional.case} can be extended for such kernel functions $f$. In this case all formulas in Theorems~\ref{thm:general.f.general.case} and \ref{thm:general.f.exceptional.case} still hold provided the index $n$ starts with $-\infty$. 
%
In particular, we note that the infinite series $\sum_{n=-\infty, n \neq -1}^{2p} a_n \zetafcn(-n) \left| \Gamma \right|^{n} N^{1-n}$ appearing in the asymptotics of $\mathcal{M}(\Gamma, f; N)$ converges for every $N$, since $\zetafcn(m) \leq \zetafcn(2)$ for all integers $m\geq2$.

Example~\ref{eg:ess.sing.1} follows from the extended version of Theorem~\ref{thm:general.f.exceptional.case}.

To justify Example~\ref{eg:ess.sing.2} let $\lambda$ be a zero of the Bessel function $\BesselJ_{-1}$. The extended version of Theorem~\ref{thm:general.f.general.case} with $a_n = \BesselJ_{n}(\lambda)$ gives that for positive integers $p \geq2$ and $m\geq2$ 
\begin{equation*}
\begin{split}
\mathcal{M}&(\Gamma, f; N) = V_f(\Gamma) \, N^2 + 2 \sum_{\substack{n=-2p, \\ n \neq \pm 1}}^{\infty} \BesselJ_{-n}(\lambda) \zetafcn(n) \left| \Gamma \right|^{-n} N^{1+n} + \mathcal{B}_p(\Gamma, f;N) + \mathcal{O}(N^{1-2p}) \\
&= 2 N \sum_{n=m}^\infty \BesselJ_{-n}(\lambda) \zetafcn(n) ( N / \left| \Gamma \right| )^n + 2 \sum_{n=2}^{m-1} \BesselJ_{-n}(\lambda) \zetafcn(n) \left| \Gamma \right|^{-n} N^{1+n} + V_f(\Gamma) \, N^2 + \left| \Gamma \right| B_2(\frac{\kappa}{2}) f^\prime( \frac{\left| \Gamma \right|}{2} ) \\
&\phantom{=}+ 2 \sum_{k=2}^{2p} J_k(\lambda) \zetafcn(-k) \left| \Gamma \right|^k N^{1-k} + \sum_{n = 2}^{p} \frac{2B_{2n}(\kappa/2)}{(2n)! \left| \Gamma \right|^{1-2n}} f^{(2n-1)}( \left| \Gamma \right| / 2 ) N^{2-2n} + \mathcal{O}(N^{1-2p}),
\end{split}
\end{equation*}
where
\begin{equation*}
V_f(\Gamma) = \frac{2}{\left| \Gamma \right|} \sum_{\substack{n=-\infty, \\ n \neq \pm 1}}^{\infty} \BesselJ_n(\lambda) \frac{\left( \left| \Gamma \right| / 2 \right)^{1+n}}{1+n}.
\end{equation*}
In the above we used the relation \eqref{eq:B.p}. 
Observe that $\zetafcn(-k)=0$ for $k=2,4,6,\dots$. 
\end{proof}

\begin{proof}[Proof of Theorem~\ref{thm:weighted.f}]
The asymptotics and the remainder estimates follow from Theorem~\ref{thm:general.f.general.case} on observing that $f_s^w(x)$ has derivatives of all orders in $(0,|\Gamma|/2+\eps)$, $S_q(x) = \sum_{n=0}^q a_n x^{n-s}$, and $\delta=1$. The constraints on $s_q = s - q$ imply that the positive integers $q,p$ and $s\in\mathbb{C}$ satisfy $q-2p < \re s < 2 + q$ or $s = q - 2p$. For $0<s<1$ we have (see \eqref{eq:general.case.V.f})
\begin{equation*}
V_{f_s^w}(\Gamma) = \frac{2}{\left| \Gamma \right|} \int_0^{\left| \Gamma \right| / 2} f_s^w(x) \dd x = \frac{2}{\left| \Gamma \right|} \sum_{n=0}^{\infty} a_n \frac{\left( \left| \Gamma \right| / 2 \right)^{1+n-s}}{1+n-s}
\end{equation*}
and the right-hand side as a function of $s$ is analytic in $\mathbb{C}$ except for poles at $s=1+n$ ($n=0,1,2,\dots$) provided $a_n\neq0$. 
\end{proof}

Using the same method of proof as in \cite{Ka1998} for the Hurwitz zeta function, we obtain the following two propositions, which will be used in the proofs of Theorems~\ref{thm.M.0} and \ref{thm:main}.

\begin{prop} \label{prop:hzeta}
Let $q \geq 1$ and $\alpha = 1 / 2$ or $\alpha = 1$. For $x > 0$ and $s \in \mathbb{C}$ with $s \neq 1$ and $\re s + 2q + 1 > 0$ the {\em Hurwitz zeta function} defined as $\zetafcn(s,a) \DEF \sum_{k=0}^\infty (k + a)^{-s}$ for $\re s > 1$ and $a\neq 0, -1, -2, \dots$ has the following representation
\begin{equation*}
\zetafcn( s, x + \alpha ) = \frac{x^{1-s}}{s-1} - B_1(\alpha) \, x^{-s} + \sum_{n=1}^{q} \frac{B_{2n}(\alpha)}{(2n)!} \Pochhsymb{s}{2n-1} x^{1-s-2n} + \rho_q(s,x,\alpha).
\end{equation*}
The remainder term is given by 
\begin{equation*}
\rho_q(s,x,\alpha) = \frac{1}{2 \pi i} \int_{\gamma_q - i \infty}^{\gamma_q + i \infty} \frac{\gammafcn(-w) \gammafcn(s+w)}{\gammafcn(s)} \zetafcn(s+w, \alpha) x^w \dd w = \mathcal{O}_{s,q}(x^{-1-\re s - 2q})
\end{equation*}
as $N \to \infty$, where $-1 - \re s - 2q < \gamma_q < - \re s - 2q$. 
\end{prop}

By the well-known relation $\log [ \gammafcn(x + \alpha) / \sqrt{2\pi} ] = \frac{\partial}{\partial s} \zetafcn( s, x + \alpha ) |_{s = 0}$
one obtains the next result from Proposition~\ref{prop:hzeta}.

\begin{prop} \label{prop:log}
Let $q \geq 1$ and $\alpha = 1 / 2$ or $\alpha = 1$. For $x > 0$
\begin{equation*}
\log \frac{\gammafcn(x + \alpha)}{\sqrt{2 \pi}} = \left( x + \alpha - 1 / 2 \right) \log x - x + \sum_{n=1}^q \frac{B_{2n}(\alpha)}{\left( 2 n - 1 \right) 2 n} x^{1-2n} + \rho_q(x,\alpha).
\end{equation*}
The remainder term is given by
\begin{equation*}
\rho_q(x,\alpha) = \frac{1}{2\pi i} \int_{\gamma_q- i \infty}^{\gamma_q+ i \infty} \gammafcn(-w) \gammafcn(w) \zetafcn(w,\alpha) x^w \dd w = \mathcal{O}_{q}(x^{-1-2q})
\end{equation*}
as $N \to \infty$, where $-1 - 2q < \gamma_q < - 2q$.
\end{prop}
In the proofs of Theorems~\ref{thm.M.0} and \ref{thm:main} we make use of the observation that for $N = 2 M + \kappa$ with $M \geq 1$ and $\kappa = 0, 1$ formula \eqref{eq:cal.M.Gamma.f.N} simplifies to
\begin{equation}
\mathcal{M}_s(\Gamma; N) = \frac{2}{\left| \Gamma \right|^s} N^{1+s} \sum_{k=1}^{\lfloor N / 2 \rfloor} \frac{1}{k^s} - \frac{1-\kappa}{\left( \left| \Gamma \right| / 2 \right)^s} N, \label{scr.M}
\end{equation}
which involves the {\em generalized harmonic numbers $H_n^{(s)} \DEF \sum_{k=1}^n k^{-s}$}.

\begin{proof}[Proof of Theorem \ref{thm.M.0}]
Differentiating \eqref{scr.M} with respect to $s$ and taking the limit $s\to0$ yields
\begin{equation*}
\mathcal{M}_{\mathrm{log}}(\Gamma; N) = N \left( N - \kappa \right) \log \frac{N}{\left| \Gamma \right|} - 2 N \log \gammafcn(\lfloor N / 2 \rfloor + 1) - \left( 1 - \kappa \right) N \log ( N / 2).
\end{equation*}
The asymptotic expansion of the theorem now follows by applying Proposition~\ref{prop:log} with $x = N / 2$, $\alpha = (2-\kappa) / 2$. Note that $B_{2n}( \alpha ) = B_{2n}( 1 - \kappa / 2 ) = B_{2n}(\kappa/2)$.
\end{proof}

\begin{proof}[Proof of Theorem \ref{thm:main}]
Starting with Theorem~\ref{thm:general.f.general.case}, we obtain an asymptotic formula of the form \eqref{gen:asympt.1} but with error estimate $\mathcal{O}(N^{1-2q})$.\footnote{If $\re s = - 2q$ and $s \neq 2q$, then a factor $\log N$ must be included.}
On the other hand, substitution of the identity $\sum_{k=1}^n k^{-s} = \zetafcn(s) - \zetafcn(s,n+1)$ into \eqref{scr.M} gives the exact formula
\begin{equation*}
\mathcal{M}_s( \Gamma; N)= \frac{2\zetafcn(s)}{\left| \Gamma \right|^s} N^{1+s} - \frac{2}{\left| \Gamma \right|^s} N^{1+s} \zetafcn(s,\lfloor N / 2 \rfloor + 1) - \frac{1-\kappa}{\left( \left| \Gamma \right| /2 \right)^s} N.
\end{equation*}
%
Then the asymptotic relation \eqref{gen:asympt.1} with error term of order $\mathcal{O}(N^{-2q})$ follows by applying Proposition~\ref{prop:hzeta} with $x = N / 2$, $\alpha = (2-\kappa) / 2$. This expansion holds for $s$ with $\re s + 2q + 1 > 0$, $q \geq 1$.
\end{proof}

\begin{proof}[Proof of Proposition~\ref{prop:M.neg.p}]
Using Jacob Bernoulli's famous closed form summation formula (\cite[Eq.~(23.1.4)]{AbSt1992}) $1^p + 2^p + \cdots + n^p = (B_{p+1}(n+1)-B_{p+1}) / ( p + 1 )$ in \eqref{scr.M} one gets
\begin{equation*}
\mathcal{M}_{-p}(\Gamma; N) = 2 \left| \Gamma \right|^p \frac{B_{p+1}((N+\kappa)/2) - B_{p+1}}{p+1} N^{1-p} + \left( 1 - \kappa \right) \left( \left| \Gamma \right| / 2 \right)^p N.
\end{equation*}
Application of the addition theorem for Bernoulli polynomials (see \cite[Eq.~(23.1.7)]{AbSt1992}) yields the result.
\end{proof}

\begin{proof}[Proof of Theorem \ref{thm:s.EQ.1}]
An asymptotic formula with error estimate $\mathcal{O}(N^{1-2p})$ follows from Theorem~\ref{thm:general.f.exceptional.case}; see also the second remark after Theorem~\ref{thm:general.f.exceptional.case}. However, by substituting into \eqref{scr.M} with $\omega=\kappa/2$ ($\kappa=0,1$) the following asymptotic expansions 
\begin{equation} \label{H.n}
H_n = \sum_{k=1}^n \frac{1}{k} = \log ( n + \omega ) + \gamma - \frac{B_1(\omega)}{n+\omega} - \sum_{k=1}^q \frac{B_{2k}(\omega) / (2k)}{\left( n + \omega \right)^{2k}} \pm \theta_{q,N,\kappa} \frac{B_{2q+2}(\omega) / (2q+2)}{\left( n + \omega \right)^{2q+2}},
\end{equation}
where $0<\theta_{q,N,\kappa}<1$ and collecting terms we get the asymptotic formula \eqref{M.s.EQ.1} with improved error estimate. The plus sign in \eqref{H.n} is taken if $\omega=1/2$ and the negative sign corresponds to $\omega=0$. We remark that the representation \eqref{H.n} is given in \cite{DeSh1991} if $\omega = 1/2$ and can be obtained as an application of the Euler-MacLaurin summation formula if $\omega = 0$ (see, for example, \cite{Ap1999}). We leave the details to the reader.
%
%
\end{proof}

\bibliographystyle{abbrv}
\bibliography{/home/jsb/APART_Fellowship/1stYEAR/PROJECTS/Bibliography/bibliography,/home/jsb/APART_Fellowship/1stYEAR/PROJECTS/Bibliography/ENERGYbibliography}

\end{document}